\documentclass[lettersize,journal]{IEEEtran}
\usepackage{epsfig,url}
\usepackage{amsmath,graphicx,subcaption,caption,amsthm,mathrsfs}
\usepackage{tikz,cite,setspace}
\usetikzlibrary{decorations.pathreplacing,calligraphy}
\usetikzlibrary{dsp,chains,arrows,shapes,spy,fit}
\usepackage{enumitem}
\tikzset{
    font={\fontsize{9}{11.0476pt}\selectfont}}
\usepackage{mathtools,amssymb,cuted,bm}
\usepackage{pgfplots}
\pgfplotsset{compat=newest}
\usepackage{romannum}
\usepackage{algorithm,algcompatible}

\newtheorem{prop}{Proposition}

\begin{document}

\title{User Activity Detection with Delay-Calibration for Asynchronous Massive Random Access}

\author{Zhichao~Shao,~\IEEEmembership{Member,~IEEE},~Xiaojun~Yuan,~\IEEEmembership{Senior~Member,~IEEE},~Rodrigo~C.~de~Lamare,~\IEEEmembership{Senior~Member,~IEEE},~and~Yong~Zhang
\thanks{This work was supported in part by the National Key Research and Development Program of China under Grant 2021YFB2900404.
	
	Z. Shao is with the Yangtze Delta Region Institute (Quzhou), University of Electronic Science and Technology of China, Quzhou 324003, China (e-mail: zhichao.shao@csj.uestc.edu.cn). X. Yuan and Y. Zhang are with the National Key Laboratory of Wireless Communications, University of Electronic Science and Technology of China, Chengdu 611731, China (e-mail: \{xjyuan,zhang.yong\}@uestc.edu.cn). Rodrigo C. de Lamare is with the Department of Electrical Engineering (DEE), Pontifical Catholic University of Rio de Janeiro, Rio de Janeiro 22451-900, Brazil (e-mail: delamare@puc-rio.br). This work was presented in part at the 2023 IEEE International Symposium on Information Theory, Taipei, Taiwan \cite{10207000}.}}

\maketitle

\pagenumbering{arabic}

\begin{abstract}
	This work considers an uplink asynchronous massive random access scenario in which a large number of users asynchronously access a base station equipped with multiple receive antennas. The objective is to alleviate the problem of massive collision due to the limited number of orthogonal preambles of an access scheme in which user activity detection is performed. We propose a user activity detection with delay-calibration (UAD-DC) algorithm and investigate the benefits of oversampling for the estimation of continuous time delays at the receiver. The proposed algorithm iteratively estimates time delays and detects active users by noting that the collided users can be identified through accurate estimation of time delays. Due to the sporadic user activity patterns, the user activity detection problem can be formulated as a compressive sensing (CS) problem, which can be solved by a modified Turbo-CS algorithm under the consideration of correlated noise samples resulting from oversampling. A sliding-window technique is applied in the proposed algorithm to reduce the overall computational complexity. Moreover, we propose a new design of the pulse shaping filter by minimizing the Bayesian Cram\'er-Rao bound of the detection problem under the constraint of limited spectral bandwidth. Numerical results demonstrate the efficacy of the proposed algorithm in terms of the normalized mean squared error of the estimated channel, the probability of misdetection and the successful detection ratio.
\end{abstract}


\section{Introduction}
\IEEEPARstart{I}{n} recent years, the widespread applications of internet of things (IoT) have experienced an explosive growth in a variety of fields \cite{6714496}, such as healthcare services \cite{6488244}, transportation and logistics \cite{ATZORI20102787}. It has been predicted that over 75.4 billion devices will be linked to the internet all over the world by 2025 \cite{9205230}, and this number will continue to increase over the next decade. This inspires the intensive study of massive random access (RA). In RA, each device needs to establish a connection with the base station (BS) before data transmission. In massive RA, a BS is required to provide connectivity to a huge number of devices, where the activity pattern of devices is typically sporadic \cite{6916986,7852531,6525600}, i.e., devices are kept in a sleep mode to save energy and only a random subset of devices are activated when triggered by external events \cite{7565189}. Since the messages for connection from different devices reach the BS at different moments, the BS needs to dynamically detect active devices and estimate their time delays for the subsequent data transmission phase.

In the existing communication systems, such as the fourth-generation Long-Term Evolution (LTE) \cite{3gpp} and the fifth-generation New Radio (NR) \cite{3gpp3}, grant-based RA is carried out based on the following four steps. In the first step, an active user transmits a randomly selected preamble over the physical random access channel (PRACH). Due to limited coordination between the BS and active users at this very early stage of communications, the transmitted preamble signals asynchronously arrive at the BS. In the second step, the BS detects the preambles transmitted by active users, determines their time delays through cross-correlation with local preambles, and sends responses to active users. Once an active user is connected to the BS successfully, it transmits data packets in the third step over the physical uplink shared channel (PUSCH). In the last step, the BS sends contention resolution notifications to active users after successfully decoding their data.

However, the above conventional RA technique will meet difficulty when applied to the scenario of massive RA. Due to the limited number of orthogonal preamble sequences, the probability of multiple active users selecting a common sequence (referred to as collisions) may increase sharply, since the number of simultaneously active users is large in massive RA. The cross-correlation method mentioned above cannot precisely detect collided active users and thus suffers from a high probability of misdetection of active users. Confronted with this problem, existing solutions are briefed as follows.

In grant-based RA, several improved ALOHA-based protocols \cite{4155680,6847724,7302046,8057803,5963225,905024,6324672} have been proposed to mitigate collisions in the literature. The contention resolution diversity slotted ALOHA \cite{4155680} and the asynchronous contention resolution diversity ALOHA \cite{6847724} transmit multiple replicas of the same preamble and use iterative interference cancellation techniques to cancel the interference caused by collided packets. Similarly, the coded slotted ALOHA proposed in \cite{7302046} combines the packet erasure correcting codes and the successive interference cancellation to resolve collisions. The enhanced contention resolution ALOHA \cite{8057803} employs selection combining and maximal-ratio combining techniques to resolve collisions when successive interference cancellation fails in the contention resolution ALOHA \cite{5963225}. Moreover, both the spread spectrum ALOHA \cite{905024} and the enhanced spread spectrum ALOHA \cite{6324672} employ direct sequence spread spectrum on the transmitted preambles to reduce collisions. 


Recently, grant-free (GF) RA \cite{8454392} was proposed to reduce the signaling overhead of the four-step procedure in grant-based RA. In GF-RA, an active user transmits a preamble over the PRACH and then data over the PUSCH in a time division multiplexing manner without waiting for a response from the BS \cite{9261952}. In general, there are two types of GF-RA: sourced and unsourced GF-RA. In sourced GF-RA, the four-step procedure is reduced to two steps, where the first and third steps are combined into one step, and the second and fourth steps are combined into the other step. The joint user activity detection and time delay estimation problem over the PRACH has been studied in \cite{9413870,9691883,9390399}, where the authors proposed various compressive sensing-based algorithms. 
After obtaining time delays, the authors in \cite{10375263,9714456,9372306} proposed several message passing algorithms to detect active users, estimate channels, and recover data jointly over the PUSCH. The unsourced GF-RA, i.e., unsourced multiple access (UMA) \cite{8006984}, has been proposed further to reduce the signaling overhead of the sourced GF-RA. In unsourced GF-RA, all the users share a common codebook for data. The receiver's task is to decode data without sending responses to active users. As a result, the two-step procedure in sourced GF-RA is reduced to one step. Existing works on unsourced GF-RA assumed perfect or coarsely perfect synchronization between BS and users \cite{9153051,9374476,9178409,9978081,li2023graphbased}, which is an oversimplified assumption.

The above-mentioned RA schemes share a common mechanism, where preambles are first sent over the PRACH for user activity detection and time delays estimation, followed by data transmission over the PUSCH. This work is focused on the user collision problem over the PRACH in asynchronous massive RA, where a large number of users share a common pool of preamble sequences. Each active user attempts to access the BS by transmitting a randomly selected preamble in an asynchronous manner, i.e., no symbol-level synchronization is achieved at the BS. The BS then detects active users and estimates their time delays based on the received preambles. The existing approaches are either inefficient or inapplicable when applied to the considered scenario of asynchronous massive RA. First, in the current LTE \cite{3gpp} and NR \cite{3gpp3} systems, preamble sequences are carried by multiple orthogonal subcarriers. In order to maintain the orthogonality of orthogonal frequency division multiplexing (OFDM), the length of the cyclic prefix (CP) must exceed the sum of the maximal multipath delay and the maximal round-trip transmission delay. In this way, the inter-symbol-interference (ISI) and the inter-carrier-interference (ICI) caused by the asynchronous transmission of preambles can be eliminated by CP removal. However, since the CP is useless in user activity detection, long CP may significantly reduce the spectrum efficiency over the PRACH. Second, the improved ALOHA-based protocols \cite{4155680,6847724,7302046,8057803,5963225,905024,6324672} as described previously either transmit multiple replicas of the same preamble or use erasure correcting codes or spread spectrum on the preambles, which consumes a large signaling overhead. Third, the asynchronous GF-RA schemes \cite{9413870,9691883,9390399} expand a preamble with multiple delayed versions in the construction of the measurement matrix. This enlarges the dimensions of the measurement matrix, and significantly increases the computational complexity of the receiver, especially when user signals arrive at the BS in a very asynchronous manner, which is exactly the case considered in this paper. More importantly, the asynchronous GF-RA schemes \cite{9413870,9691883,9390399} assume that time delays are integer multiples of symbol duration. Due to the limited coordination between the BS and the activated users at this very early stage of communications, this symbol-level synchrony between users is difficult to achieve at the BS. As such, it is highly desirable to find an efficient solution for user activity detection and time delay estimation in asynchronous massive RA.

In this paper, we propose a new solution to efficiently mitigate the user collision problem over the PRACH in asynchronous massive RA. In our proposed scheme, the BS continuously receives signals, and detects active users through calibrating their time delays. We tackle this challenging receiver design problem by developing a novel user activity detection algorithm to enable asynchronous massive RA with affordable complexity. Our main contributions are outlined as follows: 
\begin{itemize}
	\item We establish a sliding-window-based system model for uplink asynchronous massive RA, where oversampling is employed at the receiver to avoid information loss due to user asynchronism. Note that in our established model, the time delay of each active user takes a continuous value rather than an integer multiple of the symbol duration as in the prior works \cite{9413870,9691883,9390399}.
	\item We propose an expectation-maximization (EM) based user activity detection method, named user activity detection with delay-calibration (UAD-DC), to detect active users through calibrating their time delays. The novelty of our proposed method is two-fold. In the expectation step (E-step) of EM, a modified turbo compressed sensing (Turbo-CS) algorithm is developed to obtain the posterior probability via exploiting sparsity on the active users, where new treatments are required to modify the original Turbo-CS algorithm \cite{7912330} to handle the colored noise caused by oversampling and reduce the complexity of linear minimum mean square error (LMMSE) estimator. In the maximization step (M-step) of EM, a novel delay-calibration method is proposed to maximize the evidence lower bound.
	\item To improve the user activity detection performance in the considered oversampled system, we propose a new design of the pulse shaping filter by minimizing the Bayesian Cram\'er-Rao bound (BCRB) of the detection problem under the constraint of limited spectral bandwidth.
\end{itemize}
Numerical results demonstrate the superior performance of the proposed algorithm in terms of the probability of misdetection and the normalized mean square error (NMSE) of channel matrix. Our numerical experiments show that the proposed detection algorithm can increase the number of successfully detected active users by up to 38.90\%, as compared to the conventional cross-correlation-based RA.

\subsection{Organization and Notations}
The rest of this paper is organized as follows. In Section
\Romannum{2}, the system model for the asynchronous massive RA and the sliding-window-based discrete system are described. In Section \Romannum{3}, we propose the UAD-DC algorithm for each window, where the delay-calibration method is realized by an EM-based parameter learning algorithm and the user activity detection problem is solved by the modified Turbo-CS algorithm. In Section \Romannum{4}, we derive the analytical BCRB for the modified Turbo-CS algorithm. Based on the derived BCRB, we propose an optimization algorithm in Section \Romannum{5} for obtaining a new design of the pulse shaping filter to further improve the performance of user activity detection. Simulations are presented and discussed in Section \Romannum{6} and the paper is concluded in Section \Romannum{7}.

Notation: The following notation is used throughout the paper. Matrices are in bold capital letters and vectors in bold lowercase. The matrix $\mathbf{I}_n$ denotes the $n\times n$ identity matrix and the vector $\mathbf{0}_{n\times 1}$ denotes the $n\times 1$ all-zero column vector. The trace, transpose and conjugate transpose of $\mathbf{A}$ are represented by $\text{Tr}\{\mathbf{A}\}$, $\mathbf{A}^T$ and $\mathbf{A}^H$, respectively. The operators $|\cdot|$, $E\{\cdot\}$ and $\text{Var}\{\cdot\}$ calculate the absolute value, the expectation and the variance, respectively. Finally, $a(t)\ast b(t)$ and $a(t)\otimes b(t)$ denote the convolution and the Kronecker product of $a(t)$ and $b(t)$, respectively.

\section{System Model}
In this section, we describe the system model and formulate the problem to be addressed. In particular, we detail the transceiver structure for asynchronous massive RA and the sliding-window-based discrete system model.
\subsection{Transceiver Structure for Asynchronous Massive RA}
We consider a cellular system consisting of one BS equipped with $N_\text{R}$ receive antennas and several single-antenna users, as shown in Fig. \ref{fig:system_model}. While accessing the BS, each active user picks one preamble from the shared set of preambles at random.
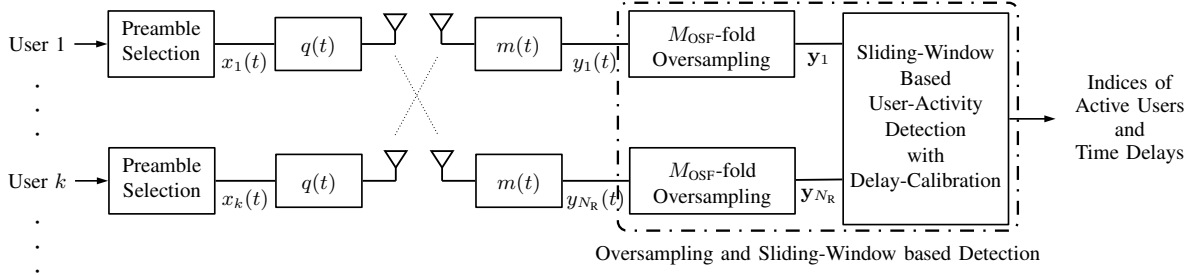
\begin{figure*}[!htbp]
	\centering
	\resizebox{0.9\textwidth}{!}{\def\antenna{
	-- +(0mm,2.0mm) -- +(1.625mm,4.5mm) -- +(-1.625mm,4.5mm) -- +(0mm,2.0mm)
}
\tikzset{%
	harddecision/.style={draw, 
		path picture={
			\pgfpointdiff{\pgfpointanchor{path picture bounding box}{north east}}%
			{\pgfpointanchor{path picture bounding box}{south west}}
			\pgfgetlastxy\x\y
			\tikzset{x=\x*.4, y=\y*.4}
			%
			\draw (-0.5,-0.5)--(0,-0.5)--(0,0.5)--(0.5,0.5);  
			\draw (-0.25,0)--(0.25,0);
	}}
}

\begin{tikzpicture}
	\node (c0) {\small User 1};
	\node[dspsquare, right= 0.5cm of c0,minimum width=1.6cm,minimum height=1cm,text height=1.7em]       (c2) {\small Preamble\\ Selection};
	\node[dspsquare, right= 0.9cm of c2,minimum width=1.3cm]       (c10) {\small $q(t)$};
	\node[coordinate,right= 0.5cm of c10] (c3) {};

	\node[below= 0.25cm of c0] (c222) {\tiny \textbullet};
	\node[below= 0.04cm of c222] (c2222) {\tiny \textbullet};
	\node[below= 0.04cm of c2222] (c22222) {\tiny \textbullet};
	\node[coordinate,below= 0.62cm of c2] (cfix) {};
	
	\node[below= 1.6cm of c0] (c00) {\small User $k$};
	\node[dspsquare, right= 0.5cm of c00,minimum width=1.6cm,minimum height=1cm,text height=1.7em]                    (c22) {\small Preamble\\ Selection};
	\node[dspsquare, right= 0.9cm of c22,minimum width=1.3cm]       (c12) {\small $q(t)$};
	\node[coordinate,right= 0.5cm of c12] (c33) {};
	
	\node[below= 0.25cm of c00] (c2225) {\tiny \textbullet};
	\node[below= 0.04cm of c2225] (c22225) {\tiny \textbullet};
	\node[below= 0.04cm of c22225] (c222225) {\tiny \textbullet};
	
	\node[coordinate,right= 0.7cm of c3] (c8) {};
	\node[coordinate,right= 0.7cm of c33] (c9) {};
	
	\node[coordinate,right= 0.5cm of c10.-15] (c88) {};
	\node[coordinate,below= 1.2cm of c88] (c99) {};
	
	\node[dspsquare, right= 0.5cm of c8,minimum width=1.3cm]       (c42) {\small $m(t)$};
	\node[dspsquare, right= 0.5cm of c9,minimum width=1.3cm]       (c43) {\small $m(t)$};
	
	\node[dspsquare, right= 1cm of c42,text width=2.5cm,minimum width=1cm,minimum height=1cm,text height=-1em]       (c13) {\small\doublespacing $M_{\text{OSF}}$-fold \\ Oversampling};
	\node[dspsquare, right= 1cm of c43,text width=2.5cm,minimum width=0.5cm,minimum height=1cm,text height=-1em]       (c14) {\small\doublespacing $M_{\text{OSF}}$-fold \\ Oversampling};
	
	\node[dspsquare,right= 10.3cm of cfix,minimum height=3.2cm,text height=6.2em,minimum width=2.5cm] (c17) {\small Sliding-Window\\ Based\\ User-Activity\\ Detection\\ with\\Delay-Calibration };
	
	\node[right= 0.7cm of c17] (c20) {\parbox{2cm}{\centering\small Indices of\\\small   Active Users\\\small and \\\small Time Delays}};
	
	\node[coordinate,right= 0.65cm of c88] (c200) {};
	\node[coordinate,right= 0.65cm of c99] (c211) {};
	
	\draw[dspconn] (c0) -- node[] {} (c2);
	\draw[thick] (c2) -- node[midway,below] {\small $x_1(t)$} (c10);
	\draw[thick] (c10) -- node[] {} (c3);
	\draw[dspconn] (c00) -- node[] {} (c22);
	\draw[thick] (c22) -- node[midway,below] {\small $x_k(t)$} (c12);
	\draw[thick] (c12) -- node[] {} (c33);
	\draw [thick] (c3) \antenna;
	\draw [thick] (c33) \antenna;
	\draw [thick] (c8) \antenna;
	\draw [thick] (c9) \antenna;
	\draw[densely dotted] (c88) -- node[] {} (c211);
	\draw[densely dotted] (c99) -- node[] {} (c200);
	\draw[thick] (c8) -- node[] {} (c42);
	\draw[thick] (c9) -- node[] {} (c43);
	\draw[thick] (c42) -- node[midway,below] {$y_1(t)$} (c13);
	\draw[thick] (c43) -- node[midway,below] {$y_{N_\text{R}}(t)$} (c14);
	\draw[thick] (c13) -- node[midway,below] {\small $\mathbf{y}_1$} (c17.138);
	\draw[thick] (c14) -- node[midway,below] {\small $\mathbf{y}_{N_\text{R}}$} (c17.-144);
	\draw[dspconn] (c17) -- node[midway,above] {} (c20);
	
	\node (l3) at (9.1,-2.5) {};
	\node (l4) at (14.5,0.3) {};
	\tikzset{blue dotted2/.style={draw=black, line width=1pt,
			dash pattern=on 1pt off 4pt on 6pt off 4pt,
			inner sep=2mm, rectangle, rounded corners}};
	\node (first dotted box) [blue dotted2, fit = (l3) (l4)] {};
	\node at (first dotted box.south) [below, inner sep=2mm] {\small Oversampling and Sliding-Window based Detection};
	\end{tikzpicture}}
	\caption{System model for the scenario of massive random access.}
	\label{fig:system_model}
\end{figure*}
The preamble signal transmitted by user $k$, denoted by $x_k(t)$, is given by
\begin{equation}
	\resizebox{\columnwidth}{!}{$x_{k}(t) = x_{k,0}\delta(t)+x_{k,1}\delta(t-T_\text{S})+\cdots+x_{k,N_{\text{PL}}-1}\delta(t-(N_{\text{PL}}-1)T_\text{S}),$}
	\label{equ_x_S}
\end{equation}
where $x_{k,n}$ is the preamble symbol at the time instant $nT_\text{S}$, $N_{\text{PL}}$ is the preamble length,
$\delta(t)$ is the Dirac delta function, and $T_\text{S}$ is the symbol duration. After going through the pulse shaping filter, denoted by $q(t)$, the preamble signal is then transmitted over the channel. 


At the BS, the preamble signal at the $n_\text{R}$-th receive antenna, denoted by $s_{n_\text{R}}(t)$, is written as
\begin{equation}\label{equ_conv}
	\begin{split}
		s_{n_\text{R}}(t)&=\sum_km(t)\ast h_{k,n_\text{R}}(t)\ast q(t)\ast x_k(t-\tau_k)\\&=\sum_kh_{k,n_\text{R}}z(t)\ast x_{k}(t-\tau_k),
	\end{split}
\end{equation}
where $m(t)$ is the matched filter,  $h_{k,n_\text{R}}(t)=h_{k,n_\text{R}}\delta(t)$ is the impulse response of the channel from user $k$ to the $n_\text{R}$-th receive antenna with $h_{k,n_\text{R}}$ being the complex gain, $\tau_k$ is the time delay of the transmission from user $k$ to the BS, $z(t)$ is the convolution of $q(t)$ and $m(t)$. It is important to notice that $\tau_k$ may not be an integer multiple of $T_\text{S}$. Moreover, the noise at the $n_\text{R}$-th receive antenna, denoted by $n_{n_\text{R}}(t)$, is written as
\begin{equation}
	n_{n_\text{R}}(t)=m(t)\ast w(t),
	\label{equ_sysnoise}
\end{equation}
where $w(t)$ is the environmental noise. 
With \eqref{equ_conv} and \eqref{equ_sysnoise}, the signal at the $n_\text{R}$-th receive antenna, denoted by $y_{n_\text{R}}(t)$, is
\begin{equation}\label{sysoverall}
	\begin{split}
		y_{n_\text{R}}(t)&=s_{n_\text{R}}(t)+n_{n_\text{R}}(t)\\&=\sum_kh_{k,n_\text{R}}z(t)\ast x_{k}(t-\tau_k)+m(t)\ast w(t).
	\end{split}
\end{equation}

The continuous signal $y_{n_\text{R}}(t)$ then goes through the oversampling module to yield the discrete samples with the sampling interval $\frac{T_\text{S}}{M_{\text{OSF}}}$, where $M_{\text{OSF}}$ is the oversampling factor. In a synchronous system, Nyquist-rate sampling is sufficient since it captures all the information. However, in an asynchronous system, we can obtain more information by oversampling the received signal $y_{n_\text{R}}(t)$. Since the time delay $\tau_k$ is continuously-valued and unknown at the BS, inaccurate estimation of $\tau_k$ may cause large ISI. Oversampling can increase the resolution in time and allow the BS to estimate $\tau_k$ more accurately. However, oversampling increases the ISI and the amount of data collected at the receiver, as shown in Fig. \ref{zt}, resulting in high computational complexity involved in signal processing. In the following, we introduce a sliding-window technique so that user activity detection is performed in each window and the overall complexity is reduced.
\begin{figure}[!htbp]
	\centering
	\resizebox{\columnwidth}{!}{\input{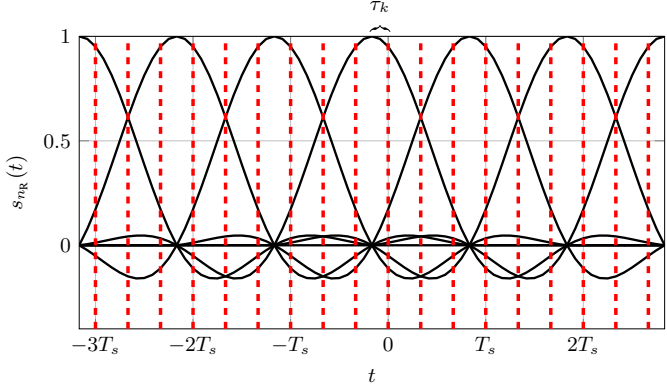}}
	\caption{An example of the ISI caused by the $k$-th user, where the 3-fold oversampling is represented by the red dashed lines, $h_{k,n_\text{R}}=1$ and $x_{k,n}=1, n\in [0,N_{\text{PL}}-1]$.}\label{zt}
\end{figure}

\subsection{Sliding-Window-Based Discrete System Model}
We now describe the sliding-window-based discrete system model as follows. An example of sliding-windows at the BS is shown in Fig. \ref{fig:transmitted_signal}, where $t$ is defined as the receiving time at the BS, $\mathbf{x}_{k}=[x_{k,0},\cdots,x_{k,N_{\text{PL}}-1}]^T\in\mathbb{C}^{N_{\text{PL}}\times 1}$ denotes the preamble sequence transmitted at user $k$, and blocks with different colors represent different preambles. Due to the limited number of orthogonal preambles, the reuse of identical preambles is unavoidable among different active users. Collisions happen when identical preambles from multiple active users overlap. The dark areas in Fig. \ref{fig:transmitted_signal} represent collisions, while the light areas represent collision-free zones. 
\begin{figure}[!htbp]
	\centering
	\resizebox{\columnwidth}{!}{\begin{tikzpicture}
	\fill [red, opacity=0.2] (1+0.3,1.5+0.5) rectangle (2.5+0.3,2+0.5);
	\node at (0+0.3,1.75+0.5) {\small Preamble 1};
	\fill [green, opacity=0.2] (4.5+0.3,1.5+0.5) rectangle (6+0.3,2+0.5);
	\node at (3.5+0.3,1.75+0.5) {\small Preamble 2};
	\fill [blue, opacity=0.2] (8+0.3,1.5+0.5) rectangle (9.5+0.3,2+0.5);
	\node at (7+0.3,1.75+0.5) {\small Preamble 3};
	
	\draw [dashed,thick] (0.5,1) -- (10.1,1) {};
	\draw [dashed,thick] (0.5,0) -- (10.1,0) {};
	\draw [dashed,thick] (0.5,-1) -- (10.1,-1) {};
	\draw [dashed,thick] (0.5,-2) -- (10.1,-2) {};
	\draw [dashed,thick] (0.5,-3) -- (10.1,-3) {};
	\draw [dashed,thick] (0.5,-4) -- (10.1,-4) {};
	\draw [dashed,thick] (0.5,-5) -- (10.1,-5) {};
	\fill [red, opacity=0.2] (3,0) rectangle (5,1);
	\fill [red, opacity=0.9] (4.8,0) rectangle (5,1);
	\node [anchor=center] at (4,0.5) {$\mathbf{x}_1$};
	\fill [red, opacity=0.9] (3,0) rectangle (3.5,1);
	\fill [red, opacity=0.9] (3,-1) rectangle (3.5,0);
	\fill [red, opacity=0.9] (4.8,-3) rectangle (5,-2);
	\fill [red, opacity=0.2] (1.5,-1) rectangle (3.5,0);
	\node [anchor=center] at (2.5,-0.5) {$\mathbf{x}_2$};
	\fill [blue, opacity=0.2] (8,-2) rectangle (10,-1);
	\node [anchor=center] at (9,-1.5) {$\mathbf{x}_3$};
	\fill [red, opacity=0.2] (4.8,-3) rectangle (6.8,-2);
	\node [anchor=center] at (5.8,-2.5) {$\mathbf{x}_4$};
	\fill [green, opacity=0.2] (7,-4) rectangle (9,-3);
	\node [anchor=center] at (8,-3.5) {$\mathbf{x}_5$};
	\fill [blue, opacity=0.2] (4,-5) rectangle (6,-4);
	\node [anchor=center] at (5,-4.5) {$\mathbf{x}_6$};
	
	\node at (-0.2,0.5) {\small User 1};
	\node at (-0.2,-0.5) {\small User 2};
	\node at (-0.2,-1.5) {\small User 3};
	\node at (-0.2,-2.5) {\small User 4};
	\node at (-0.2,-3.5) {\small User 5};
	\node at (-0.2,-4.5) {\small User 6};
	
	\fill [red, opacity=0.9] (4.5+1,-0.5-6.3+0.5) rectangle (5.5+1,-6.3+0.5);
	\node at (3.8+1,-0.25-6.3+0.5) {\small Collision};
	\fill [red, opacity=0.2] (8+0.7,-0.5-6.3+0.5) rectangle (9+0.7,-6.3+0.5);
	\node at (7+0.45,-0.25-6.3+0.5) {\small \hspace{0.5cm}No collision};
	
	\draw [<->,dotted] (0.5,0.5)--node[above]{$\tau_1$}(3,0.5);
	\draw [<->,dotted] (0.5,-0.5)--node[above]{$\tau_2$}(1.5,-0.5);
	\draw [<->,dotted] (0.5,-1.5)--node[above]{$\tau_3$}(8,-1.5);
	\draw [<->,dotted] (0.5,-2.5)--node[above]{$\tau_4$}(4.8,-2.5);
	\draw [<->,dotted] (0.5,-3.5)--node[above]{$\tau_5$}(7,-3.5);
	\draw [<->,dotted] (0.5,-4.5)--node[above]{$\tau_6$}(4,-4.5);
	
	\draw [->,thick] (0.5,-5) -- (10.1,-5) node (xaxis) [right] {$t$};
	\node at (0.5,-5.5) {0};
	\draw [-,thick,dashdotted] (0.5,-5) -- (0.5,1);
	
	\node (l1) at (0.55,-5) {};
	\node (l2) at (3.7,1) {};
	\tikzset{blue dotted/.style={draw=black!50!white, line width=1pt, dash pattern=on 1pt off 4pt on 6pt off 4pt,
			inner sep=0.3mm, rectangle, rounded corners}};
	\node (first dotted box) [blue dotted, fit = (l1) (l2)] {};
	
	\node (l3) at (2.2,-5) {};
	\node (l4) at (5.8,1) {};
	\tikzset{blue dotted2/.style={draw=blue, line width=1pt,
			dash pattern=on 1pt off 4pt on 6pt off 4pt,
			inner sep=2mm, rectangle, rounded corners}};
	\node (first dotted box) [blue dotted2, fit = (l3) (l4)] {};
	\node at (first dotted box.north) [above, inner sep=1mm] {\small Sliding-Window};
	
	\node (l5) at (6.7,-5.3) {$\cdots$};
	\draw [<->,dotted] (0.5,-6.3+0.5)--node[below]{$\Delta t$}(1.9,-6.3+0.5);
	\draw [<->,dotted] (0.5,-6.8+0.5)--node[below]{$LT_\text{S}$}(3.9,-6.8+0.5);
\end{tikzpicture}}
	\caption{An example of sliding-windows at the BS.}
	\label{fig:transmitted_signal}
\end{figure}
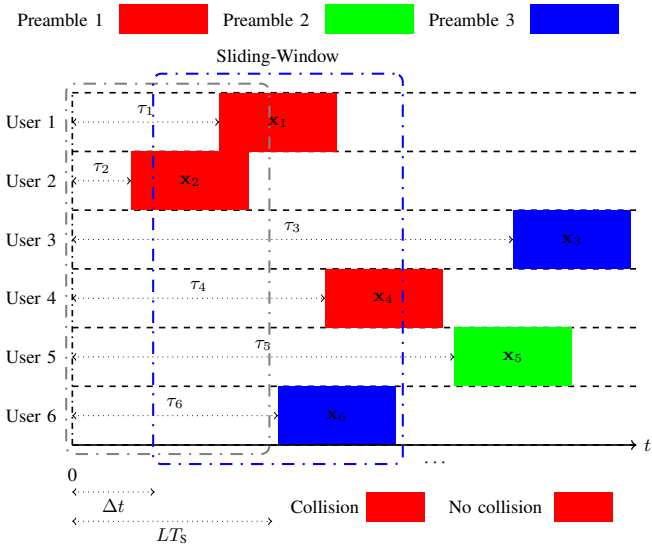

As shown in Fig. \ref{fig:transmitted_signal}, the BS detects the active users over a moving time window with the window width $LT_\text{S}$ and the stepsize $\Delta t$. The received preamble sequences $\mathbf{x}_{k}$ can be categorized into the following four types based on their position relative to the observation window, where $t_l$ is the starting time of the $l$-th window:  
\begin{itemize}
	\item  Type \uppercase\expandafter{\romannumeral1}: $\mathbf{x}_{k}$ is within the observation window, i.e., $t_l\leq \tau_{k} \leq t_l+(L-N_{\text{PL}})T_\text{S}$;
	\item  Type \uppercase\expandafter{\romannumeral2}: $\mathbf{x}_{k}$ is across the left boundary of the observation window, i.e., $t_l-N_{\text{PL}}T_\text{S}\leq\tau_{k}<t_l$;
	\item  Type \uppercase\expandafter{\romannumeral3}: $\mathbf{x}_{k}$ is across the right boundary of the observation window, i.e., $t_l+(L-N_{\text{PL}})T_\text{S}< \tau_{k}\leq t_l+LT_\text{S}$;
	\item  Type \uppercase\expandafter{\romannumeral4}: $\mathbf{x}_{k}$ is outside the observation window, i.e., $t_l-N_{\text{PL}}T_\text{S}>\tau_{k}$ or $\tau_{k}>t_l+LT_\text{S}$.
\end{itemize}
For the highlighted sliding-window, $\mathbf{x}_{1}$ and $\mathbf{x}_{6}$ belong to Type \uppercase\expandafter{\romannumeral1}, $\mathbf{x}_{2}$ belongs to Type \uppercase\expandafter{\romannumeral2}, $\mathbf{x}_{4}$ belongs to Type \uppercase\expandafter{\romannumeral3}, $\mathbf{x}_{3}$ and $\mathbf{x}_{5}$ belong to Type \uppercase\expandafter{\romannumeral4}. The preambles in Type \uppercase\expandafter{\romannumeral2} and Type \uppercase\expandafter{\romannumeral3} may result in the misdetection of the corresponding users in the current window. In order to prevent the detection of the same active user in multiple windows, we choose the stepsize
\begin{equation}
	\Delta t = (L-N_{\text{PL}})T_\text{S}.
\end{equation}

After the oversampling module, the received preamble signal in \eqref{equ_conv} in the $l$-th window is sampled as 
\begin{equation}\label{discret}
	\begin{split}
		&s_{{n_\text{R}}_l}\left[\frac{i}{M_{\text{OSF}}}\right]=\sum_{k_l=1}^{K_l}\sum_{j=0}^{LM_{\text{OSF}}-1}\alpha_{k_l}h_{k_l,n_\text{R}}z\left[\frac{i-j}{M_{\text{OSF}}}\right]\\&\qquad\times x_{k_l}\left[\frac{j}{M_{\text{OSF}}}\right],\quad
			i=0,1,\cdots,LM_{\text{OSF}}-1
	\end{split}
\end{equation}
where $k_l$ and $K_l$ denote the user $k$ and the number of users in the $l$-th window, respectively,  $s_{{n_\text{R}}_l}\left[\frac{i}{M_{\text{OSF}}}\right]=s_{n_\text{R}}(t)|_{t=t_l+\frac{i}{M_{\text{OSF}}}T_\text{S}}$,  $z\left[\frac{i-j}{M_{\text{OSF}}}\right]=z(t)|_{t=\frac{i-j}{M_{\text{OSF}}}T_\text{S}}$ \cite{9354165}. Note that $K_l$ is unknown at the receiver and is determined in the initialization step of Algorithm \ref{alg_jdcuad}. The variable $\alpha_{k_l}$ is the user activity indicator defined as
\begin{equation}
	\alpha_{k_l}=\begin{cases}
		1,\text{\quad if user $k$ is active in the $l$-th window,}\\
		0,\text{\quad otherwise.}
	\end{cases}
\end{equation} 
Since most of the energy of $z(t)$ is concentrated in a finite range, we consider $z\left[\frac{i-j}{M_{\text{OSF}}}\right]$ in the range of $\frac{i-j}{M_{\text{OSF}}}\in[-3,3]$, and omit $z\left[\frac{i-j}{M_{\text{OSF}}}\right]$ out of this range. Moreover, based on the type of received preamble, we have
	\begin{itemize}
		\item Type \uppercase\expandafter{\romannumeral1}:\\
		 \resizebox{0.85\columnwidth}{!}{$x_{k_l}\left[\frac{j}{M_{\text{OSF}}}\right]=\begin{cases}
			x_{k,n}, \text{ if }\frac{j}{M_{\text{OSF}}} =n\in [0,N_{\text{PL}}-1]\\
			0,\text{\qquad otherwise}
		\end{cases}$}
		\item Type \uppercase\expandafter{\romannumeral2}:\\ \resizebox{0.93\columnwidth}{!}{$x_{k_l}\left[\frac{j}{M_{\text{OSF}}}\right]=\begin{cases}
			x_{k,n}, \text{ if }\frac{j}{M_{\text{OSF}}} =n\in [\lfloor\frac{\tau_{k_l}}{T_\text{S}}\rfloor,N_{\text{PL}}-1]\\
			0,\text{\qquad otherwise}
		\end{cases}$}
		\item Type \uppercase\expandafter{\romannumeral3}:\\ \resizebox{0.93\columnwidth}{!}{$x_{k_l}\left[\frac{j}{M_{\text{OSF}}}\right]=\begin{cases}
			x_{k,n}, \text{ if }\frac{j}{M_{\text{OSF}}} =n\in [0,L+\lfloor\frac{\tau_{k_l}}{T_\text{S}}\rfloor]\\
			0,\text{\qquad otherwise}
		\end{cases}$}
		\item Type \uppercase\expandafter{\romannumeral4}: $x_{k_l}\left[\frac{j}{M_{\text{OSF}}}\right]=0$
\end{itemize}
where $\lfloor\cdot\rfloor$ denotes the operation of rounding to the nearest integer less than or equal to the input variable, $\tau_{k_l}=t_l-\tau_k$. The vector form of \eqref{discret} is written as
\begin{equation}
	\mathbf{s}_{{n_\text{R}}_l}= \sum_{k_l=1}^{K_l}\alpha_{k_l}h_{k_l,n_\text{R}}\mathbf{Z}\mathbf{x}_{k_l}(\tau_{k_l}),
	\label{equ_matrix_sing}
\end{equation}
where $\mathbf{s}_{{n_\text{R}}_l}=[s_{{n_\text{R}}_l}[0],\cdots,s_{{n_\text{R}}_l}[L-\frac{1}{M_{\text{OSF}}}]]^T\in\mathbb{C}^{LM_{\text{OSF}}\times1}$, $\mathbf{x}_{k_l}(\tau_{k_l})=[x_{k_l}\left[0\right],\cdots,x_{k_l}[L-\frac{1}{M_{\text{OSF}}}]]^T\in\mathbb{C}^{LM_{\text{OSF}}\times1}$. The matrix $\mathbf{Z}\in\mathbb{R}^{LM_{\text{OSF}} \times LM_{\text{OSF}}}$ is a Toeplitz matrix constructed by $z(t)$ at different time instants with the following form
\begin{equation}
	\resizebox{\columnwidth}{!}{$\mathbf{Z} = \begin{bmatrix}
		z[0] & z[-\frac{1}{M_{\text{OSF}}}] & \dots & z[-3] & 0 & 0 & \dots & 0 & 0\\ 
		z[\frac{1}{M_{\text{OSF}}}] & z[0] & \dots & z[-3+\frac{1}{M_{\text{OSF}}}] & z[-3] & 0 & \dots & 0 & 0\\
		\vdots & \vdots & \ddots & \vdots& \vdots& \vdots& \ddots& \vdots& \vdots\\
		0 & 0 & \dots & 0& 0 & z[3] & \dots & z[\frac{1}{M_{\text{OSF}}}] & z[0]\\
	\end{bmatrix}.$}
	\label{eq_zform}
\end{equation}
Assuming that $z(t)$ is a symmetric filter with $z[\frac{i-j}{M_{\text{OSF}}}]=z[\frac{j-i}{M_{\text{OSF}}}]$, $\mathbf{Z}$ is a symmetric matrix.
Similar to \eqref{discret}, the discrete-time form of \eqref{equ_sysnoise} in the $l$-th window is
\begin{equation}\label{discretnoise}
	\begin{split}
		n_{{n_\text{R}}_l}\left[\frac{i}{M_{\text{OSF}}}\right]&=\sum_{j=0}^{LM_{\text{OSF}}-1}m\left[\frac{i-j}{M_{\text{OSF}}}\right]w_l\left[\frac{j}{M_{\text{OSF}}}\right],\\&\qquad\qquad\quad
		i=0,1,\cdots,LM_{\text{OSF}}-1
	\end{split}
\end{equation}
where $w_l\left[\frac{j}{M_{\text{OSF}}}\right]=w(t)|_{t=t_l+\frac{j}{M_{\text{OSF}}}T_\text{S}}$. The vector form of \eqref{discretnoise} is written as
\begin{equation}
	\mathbf{n}_{{n_\text{R}}_l}=\mathbf{F}\mathbf{w}_{{n_\text{R}}_l},
	\label{equ_noise}
\end{equation}
where $\mathbf{F}\in\mathbb{R}^{LM_{\text{OSF}} \times LM_{\text{OSF}}}$ is the Toeplitz matrix constructed by $m(t)$ at different time instants with the form similar to \eqref{eq_zform}, and $\mathbf{w}_{{n_\text{R}}_l}\sim \mathcal{CN}(\mathbf{0}_{LM_{\text{OSF}}\times1},\sigma_\text{N}^2\mathbf{I}_{LM_{\text{OSF}}})$ represents the complex Gaussian random variables with zero mean and variance $\sigma_\text{N}^2$. With \eqref{equ_matrix_sing} and \eqref{equ_noise}, the samples at the $n_\text{R}$-th receive antenna in the $l$-th window, denoted by $\mathbf{y}_{{n_\text{R}}_l}\in\mathbb{C}^{LM_{\text{OSF}}\times 1}$, are
\begin{equation}
	\begin{split}
		\mathbf{y}_{{n_\text{R}}_l}=\sum_{k_l=1}^{K_l}\alpha_{k_l}h_{k_l,n_\text{R}}\mathbf{Z}\mathbf{x}_{k_l}(\tau_{k_l})+\mathbf{F}\mathbf{w}_{{n_\text{R}}_l}.
	\end{split}
	\label{equ_siso}
\end{equation}
Considering all the receive antennas, we obtain 
\begin{equation}\label{equ_sys1}
	\mathbf{Y}_l=[\mathbf{y}_{1_l},\cdots,\mathbf{y}_{{N_\text{R}}_l}]=
	\mathbf{Z}\mathbf{X}(\bm{\tau}_l)\mathbf{G}_l + \mathbf{FW}_l,
\end{equation}
where $\mathbf{Y}_l\in\mathbb{C}^{LM_{\text{OSF}}\times N_\text{R}}$ denote the received samples, $\mathbf{X}(\bm{\tau}_l)=[\mathbf{x}_{1_l}(\tau_{1_l}),\cdots,\mathbf{x}_{K_l}(\tau_{K_l})]\in\mathbb{C}^{LM_{\text{OSF}}\times K_l}$ is a delayed preamble matrix with continuously-valued delay $\bm{\tau}_l=[\tau_{1_l},\cdots,\tau_{K_l}]^T\in\mathbb{R}^{K_l\times 1}$, $\mathbf{G}_l=\begin{bmatrix}
	\alpha_{1_l}h_{1_l,1},\cdots,\alpha_{1_l}h_{1_l,N_\text{R}}\\ \vdots \\\alpha_{K_l}h_{K_l,1},\cdots, \alpha_{K_l}h_{K_l,N_\text{R}}
\end{bmatrix}\in\mathbb{C}^{K_l\times N_\text{R}}$ is a row sparse equivalent channel matrix,  $\mathbf{W}_l=[\mathbf{w}_{1_l},\cdots,\mathbf{w}_{{n_\text{R}}_l}]\in\mathbb{C}^{LM_{\text{OSF}}\times N_\text{R}}$ is the noise matrix. The $k$-th row of $\mathbf{G}_l$, denoted by  $\mathbf{g}_{k_l}\in\mathbb{C}^{N_\text{R}\times1}$, follows a Bernoulli Gaussian distribution:
\begin{equation}\label{equ_gdis}
	p(\mathbf{g}_{k_l})=(1-\rho_l)\delta(\mathbf{g}_{k_l})+\rho_l \mathcal{CN}(\bm{0}_{N_\text{R}\times1},\gamma_k\mathbf{I}_{N_\text{R}}),
\end{equation}
where $\rho_l$ is the probability of active users in the $l$-th window, and $\gamma_k$ is the path-loss and shadowing component depending on the user location.

\subsection{Problem Formulation}
Our target is to estimate $\bm{\tau}_l$ and $\mathbf{G}_l$ from the noisy observation $\mathbf{Y}_l$ in \eqref{equ_sys1}, where active users can be detected from the non-zero rows of estimated $\hat{\mathbf{G}}_l$. The maximum likelihood estimate of $\bm{\tau}_l$  is given as
\begin{equation}\label{equ_ml}
		\hat{\bm{\tau}}_l=\arg\max_{\bm{\tau}_l}\ln\int_{\mathbf{G}_l }p(\mathbf{G}_l,\mathbf{Y}_l;\bm{\tau}_l),
\end{equation}
where $p(\mathbf{G}_l,\mathbf{Y}_l;\bm{\tau}_l)$ denotes the joint density of $\mathbf{G}_l$ and $\mathbf{Y}_l$ with the variable $\bm{\tau}_l$. However, it is not a trivial task to solve the above problem since $\mathbf{G}_l$ is unknown. From the Jensen’s inequality \cite{boyd_vandenberghe_2004}, the log-likelihood in \eqref{equ_ml} is lower bounded by
\begin{subequations}\label{equ_mlpro}
	\begin{align}
		\ln\int_{\mathbf{G}_l }p(\mathbf{G}_l,\mathbf{Y}_l;\bm{\tau}_l)&=\ln E_{p'(\mathbf{G}_l)}\left\{\frac{ p(\mathbf{G}_l,\mathbf{Y}_l;\bm{\tau}_l)}{p'(\mathbf{G}_l)}\right\}\\&\geq E_{p'(\mathbf{G}_l)}\left\{\ln\frac{ p(\mathbf{G}_l,\mathbf{Y}_l;\bm{\tau}_l)}{p'(\mathbf{G}_l)}\right\},
	\end{align}
\end{subequations}
where $p'(\mathbf{G}_l)$ is any distribution related to $\mathbf{G}_l$. $E_{p'(\mathbf{G}_l)}\{\cdot\}$ represents the expectation over the distribution $p'(\mathbf{G}_l)$. The equality of \eqref{equ_mlpro} holds when $\ln\frac{ p(\mathbf{G}_l,\mathbf{Y}_l;\bm{\tau}_l)}{p'(\mathbf{G}_l)}$ is an affine function. This is achieved for $p'(\mathbf{G}_l)=p(\mathbf{G}_l|\mathbf{Y}_l;\bm{\tau}_l)$.
Thus, \eqref{equ_ml} is equivalent to 
\begin{equation}\label{equ_em}
		\hat{\bm{\tau}}_l=\arg\max_{\bm{\tau}_l}E_{p(\mathbf{G}_l|\mathbf{Y}_l;\bm{\tau}_l)}\left\{\ln\frac{ p(\mathbf{G}_l,\mathbf{Y}_l;\bm{\tau}_l)}{p(\mathbf{G}_l|\mathbf{Y}_l;\bm{\tau}_l)}\right\}.
\end{equation}
The posterior density $p(\mathbf{G}_l|\mathbf{Y}_l;\bm{\tau}_l)$ is difficult to obtain since $\bm{\tau}_l$ is the variable to be estimated. In the next section, we propose a low-complexity EM-based framework to solve the above problem. 
With the estimated $\hat{\bm{\tau}}_l$, the minimum mean square error (MMSE) estimator of $\mathbf{G}_l$ is given by \cite{kay}
\begin{equation}
		\hat{\mathbf{G}}_l=E\{\mathbf{G}_l|\mathbf{Y}_l\}=\int_{\mathbf{G}_l} \mathbf{G}_lp(\mathbf{G}_l|\mathbf{Y}_l;\hat{\bm{\tau}}_l).
	\label{equ_MMSE}
\end{equation}
In the following, we propose an EM-based UAD-DC method to detect active users through calibrating their time delays. For notational brevity we henceforth omit the sliding-window index $l$ in places without introducing ambiguity.

\section{User Activity Detection with Delay-Calibration}\label{sec2}
This section provides a detailed description of the proposed UAD-DC algorithm.
\subsection{EM-Based Framework}\label{sec2em}
Following the EM principle \cite{em}, UAD-DC consists of two steps named the E-step and the M-step. The E-step calculates the expectation of the log-likelihood $E_{p(\mathbf{G}|\mathbf{Y};\hat{\bm{\tau}}(u))}\left\{\ln\frac{ p(\mathbf{G},\mathbf{Y};\bm{\tau})}{p(\mathbf{G}|\mathbf{Y};\hat{\bm{\tau}}(u))}\right\}$ based on the posterior density $p(\mathbf{G}|\mathbf{Y};\hat{\bm{\tau}}(u))$. The M-step computes $\hat{\bm{\tau}}(u+1)$ for maximizing the expectation in the E-step, where $\hat{\bm{\tau}}(u)$ is the estimate of $\bm{\tau}$ in the $u$-th EM iteration.

The E-step in the $u$-th EM iteration is
\begin{equation}
	\begin{split}\label{equ_expec} E_{p(\mathbf{G}|\mathbf{Y};\hat{\bm{\tau}}(u))}\{\ln p(\mathbf{Y},\mathbf{G};\bm{\tau})\}&=\int_{\mathbf{G}} p(\mathbf{G}|\mathbf{Y};\hat{\bm{\tau}}(u))\\&\hspace{-3.5cm}\times\ln p(\mathbf{Y}|\mathbf{G};\bm{\tau}) +\int_{\mathbf{G}} p(\mathbf{G}|\mathbf{Y};\hat{\bm{\tau}}(u))\ln p(\mathbf{G}),
	\end{split}
\end{equation}
where the posterior density $p(\mathbf{G}|\mathbf{Y};\hat{\bm{\tau}}(u))$ is calculated by a modified Turbo-CS algorithm described in Section \ref{sec_delay}, and $E_{p(\mathbf{G}|\mathbf{Y};\hat{\bm{\tau}}(u))}\{\ln p(\mathbf{G}|\mathbf{Y};\hat{\bm{\tau}}(u))\}$ is dropped since it is irrelevant to $\bm{\tau}$. Assume that $p(\mathbf{g}_{n_\text{R}}|\mathbf{y}_{n_\text{R}};\hat{\bm{\tau}}(u))$ obtained from Section \ref{sec_delay} has the form of  $\mathcal{CN}\left(\mathbf{g}_{n_\text{R}};\hat{\mathbf{g}}_{n_\text{R}}(u),\mathbf{V}^{\mathbf{G}}_{n_\text{R}}(u)\right)$, where $\mathbf{g}_{n_\text{R}}$ and $\mathbf{y}_{n_\text{R}}$ are the $n_\text{R}$-th columns of $\mathbf{G}$ and $\mathbf{Y}$, respectively. The first term in \eqref{equ_expec} is calculated as
\begin{equation}\label{equ_max2}
	\resizebox{\columnwidth}{!}{$
	\begin{aligned}
		&\int_{\mathbf{G}}p(\mathbf{G}|\mathbf{Y};\hat{\bm{\tau}}(u))\ln p(\mathbf{Y}|\mathbf{G};\bm{\tau})\\&=\sum_{n_\text{R}=1}^{N_\text{R}}
		\int_{\mathbf{g}_{n_\text{R}}}p(\mathbf{g}_{n_\text{R}}|\mathbf{y}_{n_\text{R}};\hat{\bm{\tau}}(u))\ln\mathcal{CN}(\mathbf{y}_{n_\text{R}};\mathbf{Z}\mathbf{X}(\bm{\tau})\mathbf{g}_{n_\text{R}},\sigma_\text{N}^2\mathbf{F}\mathbf{F}^H)\\
		&=-N_\text{R}\ln(\pi\sigma_\text{N}^2)-\frac{1}{\sigma_\text{N}^2}\sum_{n_\text{R}=1}^{N_\text{R}}\int_{\mathbf{g}_{n_\text{R}}}\mathcal{CN}\left(\mathbf{g}_{n_\text{R}};\hat{\mathbf{g}}_{n_\text{R}}(u),\mathbf{V}^{\mathbf{G}}_{n_\text{R}}(u)\right)\\&\times\left(\mathbf{y}_{n_\text{R}}^H(\mathbf{F}\mathbf{F}^H)^{-1}\mathbf{y}_{n_\text{R}}-\mathbf{y}_{n_\text{R}}^H(\mathbf{F}\mathbf{F}^H)^{-1}\mathbf{Z}\mathbf{X}(\bm{\tau})\mathbf{g}_{n_\text{R}}-(\mathbf{Z}\mathbf{X}(\bm{\tau})\mathbf{g}_{n_\text{R}})^H\right.\\&\left.\times(\mathbf{F}\mathbf{F}^H)^{-1}\mathbf{y}_{n_\text{R}}+(\mathbf{Z}\mathbf{X}(\bm{\tau})\mathbf{g}_{n_\text{R}})^H(\mathbf{F}\mathbf{F}^H)^{-1}(\mathbf{Z}\mathbf{X}(\bm{\tau})\mathbf{g}_{n_\text{R}})\right),
	\end{aligned}$}
\end{equation}
and the last term in \eqref{equ_expec} is omitted since it is irrelevant to $\bm{\tau}$.

The M-step in the $u$-th EM iteration is
\begin{equation}\label{equ_maxi}
	\hat{\bm{\tau}}(u+1)=\arg\max_{\bm{\tau}}E_{p(\mathbf{G}|\mathbf{Y};\hat{\bm{\tau}}(u))}\{\ln p(\mathbf{Y},\mathbf{G};\bm{\tau})\}.
\end{equation}
Inserting \eqref{equ_max2} into \eqref{equ_maxi} and ignoring the terms independent of $\bm{\tau}$, \eqref{equ_maxi} is equivalent to \eqref{equ_optmax}.
\begin{figure*}[!h]
	\centering
       \begin{subequations}\label{equ_optmax}
		\begin{align}\nonumber&\hat{\bm{\tau}}(u+1)=\arg\max_{\bm{\tau}}\sum_{n_\text{R}=1}^{N_\text{R}}\int_{\mathbf{g}_{n_\text{R}}}\frac{1}{\sigma_\text{N}^2}\mathcal{CN}\left(\mathbf{g}_{n_\text{R}};\hat{\mathbf{g}}_{n_\text{R}}(u),\mathbf{V}^{\mathbf{G}}_{n_\text{R}}(u)\right)\left(\mathbf{y}_{n_\text{R}}^H(\mathbf{F}\mathbf{F}^H)^{-1}\mathbf{Z}\mathbf{X}(\bm{\tau})\mathbf{g}_{n_\text{R}}+(\mathbf{Z}\mathbf{X}(\bm{\tau})\mathbf{g}_{n_\text{R}})^H(\mathbf{F}\mathbf{F}^H)^{-1}\mathbf{y}_{n_\text{R}}\right.\\&\hspace{4.5cm}\left.-(\mathbf{Z}\mathbf{X}(\bm{\tau})\mathbf{g}_{n_\text{R}})^H(\mathbf{F}\mathbf{F}^H)^{-1}(\mathbf{Z}\mathbf{X}(\bm{\tau})\mathbf{g}_{n_\text{R}})\right)\\	&\nonumber\hspace{-0.5cm}=\arg\max_{\bm{\tau}}\frac{1}{\sigma_\text{N}^2}\sum_{n_\text{R}=1}^{N_\text{R}}2Re\left\{\mathbf{y}_{n_\text{R}}^H(\mathbf{F}\mathbf{F}^H)^{-1}\mathbf{Z}\mathbf{X}(\bm{\tau})\hat{\mathbf{g}}_{n_\text{R}}(u)\right\}-\text{Tr}\{(\mathbf{Z}\mathbf{X}(\bm{\tau}))^H(\mathbf{F}\mathbf{F}^H)^{-1}\mathbf{Z}\mathbf{X}(\bm{\tau})\int_{\mathbf{g}_{n_\text{R}}}\mathbf{g}_{n_\text{R}}\mathbf{g}_{n_\text{R}}^H\mathcal{CN}(\mathbf{g}_{n_\text{R}};\hat{\mathbf{g}}_{n_\text{R}}(u),\mathbf{V}^{\mathbf{G}}_{n_\text{R}}(u))\}\\\label{equ_objc}		&\hspace{-0.5cm}=\arg\max_{\bm{\tau}}\frac{1}{\sigma_\text{N}^2}\sum_{n_\text{R}=1}^{N_\text{R}}2Re\left\{\mathbf{y}_{n_\text{R}}^H(\mathbf{F}\mathbf{F}^H)^{-1}\mathbf{Z}\mathbf{X}(\bm{\tau})\hat{\mathbf{g}}_{n_\text{R}}(u)\right\}-\text{Tr}\left\{(\mathbf{Z}\mathbf{X}(\bm{\tau}))^H(\mathbf{F}\mathbf{F}^H)^{-1}\mathbf{Z}\mathbf{X}(\bm{\tau})\left(\hat{\mathbf{g}}_{n_\text{R}}(u)\hat{\mathbf{g}}_{n_\text{R}}(u)^H+\mathbf{V}^{\mathbf{G}}_{n_\text{R}}(u)\right)\right\}.
			\end{align}
	\end{subequations}
	\hrulefill
\end{figure*}
Obtaining a closed-form solution to \eqref{equ_objc} is challenging due to the non-linear relationship between $\mathbf{X}$ and $\bm{\tau}$. Instead, a sub-optimal solution is proposed in Algorithm \ref{alg:receiver3}, where $\mathcal{G}(\bm{\tau})$ represents the objective function in \eqref{equ_objc}. Algorithm \ref{alg:receiver3} updates the entries in $\hat{\bm{\tau}}(u+1)$ one by one until all $K$ delays are updated, where the update of the $k$-th entry, denoted by $\hat{\tau}_k(u+1)$, depends on the previously $(k-1)$ updated entries. Steps 2-3 search $\hat{\tau}_k(u+1)$ in the range of $[\hat{\tau}_k(u)-\epsilon,\hat{\tau}_k(u)+\epsilon]$ with the stepsize $\frac{T_\text{S}}{\kappa}$ , where $\epsilon$ is a pre-defined positive number that is set small enough to reduce the searching cost, and $\kappa$ is an integer greater than $M_{\text{OSF}}$ to refine the search.
\begin{algorithm}[!htbp]
	\caption{Greedy search based delay calibration}
	\begin{algorithmic}[1]
		\STATEx \textbf{Input: }$\mathbf{Y},\mathbf{Z},\mathbf{F},\hat{\bm{\tau}}(u),\hat{\mathbf{g}}_{n_\text{R}}(u), \mathbf{V}^{\mathbf{G}}_{n_\text{R}}(u)$
		\FOR{$k=1:K$}
		\STATE $\hat{\bm{\tau}}_{\text{tmp}}=[\cdots,\hat{\tau}_{k-1}(u+1),\hat{\tau}_k(u),\hat{\tau}_{k+1}(u),\cdots]$
		\STATE $\hat{\tau}_k(u+1)=\arg\max_{\hat{\tau}_k(u)-\epsilon\leq\hat{\tau}_k(u)\leq\hat{\tau}_k(u)+\epsilon}\mathcal{G}(\hat{\bm{\tau}}_{\text{tmp}})$
		\ENDFOR
		\STATEx \textbf{Output: }$\hat{\bm{\tau}}(u+1)=\hat{\bm{\tau}}_{\text{tmp}}$
	\end{algorithmic}
	\label{alg:receiver3}
\end{algorithm}	

\subsection{Compressive Sensing with Colored Noise}\label{sec_delay}
With fixed $\hat{\bm{\tau}}(u)$, the posterior density $p(\mathbf{G}|\mathbf{Y};\hat{\bm{\tau}}(u))$ is calculated as 
\begin{equation}
	p(\mathbf{G}|\mathbf{Y};\hat{\bm{\tau}}(u))=\frac{p(\mathbf{Y}|\mathbf{G};\hat{\bm{\tau}}(u))p(\mathbf{G})}{p(\mathbf{Y};\hat{\bm{\tau}}(u))}.
\end{equation}	
Combining the likelihood $p(\mathbf{Y}|\mathbf{G};\hat{\bm{\tau}}(u))$ and the prior density $p(\mathbf{G})$,
several existing sparse recovery algorithms can be applied to estimate $p(\mathbf{G}|\mathbf{Y};\hat{\bm{\tau}}(u))$, such as Bernoulli-Gaussian generalized approximate message (BG-GAMP) \cite{6556987} and vector approximate message passing (VAMP) \cite{8713501}. However, the above algorithms cannot be directly applied due to oversampling. This leads to the following two challenges: first, the entries in $\mathbf{Z}\mathbf{X}(\hat{\bm{\tau}}(u))$ are not independently and identically distributed (i.i.d.); second, the filtered noise samples $\mathbf{F}\mathbf{W}$ are correlated with non-diagonal covariance  $\sigma_\text{N}^2\mathbf{F}\mathbf{F}^H$.
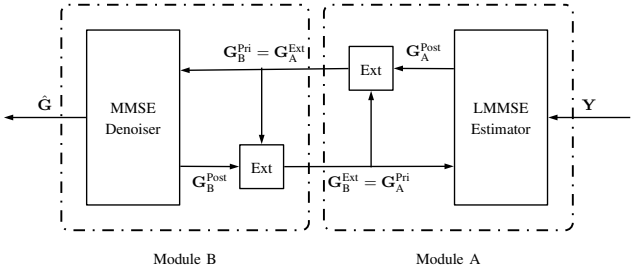
\begin{figure}[!htbp]
	\centering
	\resizebox{\columnwidth}{!}{\tikzset{%
	harddecision/.style={draw, 
		path picture={
			\pgfpointdiff{\pgfpointanchor{path picture bounding box}{north east}}%
			{\pgfpointanchor{path picture bounding box}{south west}}
			\pgfgetlastxy\x\y
			\tikzset{x=\x*.4, y=\y*.4}
			%
			\draw (-0.5,-0.5)--(0,-0.5)--(0,0.5)--(0.5,0.5);  
			\draw (-0.25,0)--(0.25,0);
	}}
}

\begin{tikzpicture}
	\node (c3) at (0,0) {};
	\node[dspsquare,right= 1.5cm of c3,minimum height=3.2cm,text height=2em,minimum width=1.7cm] (c1) {\footnotesize MMSE \\ \footnotesize Denoiser};
	\node (c4) at (0,0.9) {};
	\node[dspsquare,right= 6.3cm of c4] (c5) {\footnotesize Ext};
	\node (c6) at (0,-0.9) {};
	\node[dspsquare,right= 4.3cm of c6] (c7) {\footnotesize Ext};
	\node[dspsquare,right= 5cm of c1,minimum height=3.2cm,text height=2em,minimum width=1.7cm] (c2) {\footnotesize LMMSE \\ \footnotesize Estimator};
	
	\draw[dspconn] (c1) -- node[midway,above] {\footnotesize $\hat{\mathbf{G}}$} (c3);
	\draw[dspconn] (c5) -- node[midway,above] {\footnotesize $\mathbf{G}_\text{B}^{\text{Pri}}=\mathbf{G}_\text{A}^{\text{Ext}}$} (c1.45);
	\draw[dspconn] (c1.-46) -- node[midway,below] {\footnotesize $\mathbf{G}_\text{B}^{\text{Post}}$} (c7);
	\draw[dspconn] (c2.134) -- node[midway,above] {\footnotesize $\mathbf{G}_\text{A}^{\text{Post}}$} (c5);
	\draw[dspconn] (c7) -- node[midway,below] {\footnotesize $\mathbf{G}_\text{B}^{\text{Ext}}=\mathbf{G}_\text{A}^{\text{Pri}}$} (c2.-134);
	
	\node[below= 1.4cm of c5] (c10) {};
	\draw[dspconn] (c10) -- node[midway,above] {} (c5);
	
	\node[above= 1.4cm of c7] (c11) {};
	\draw[dspconn] (c11) -- node[midway,above] {} (c7);
	
	\node[right= 1.5cm of c2] (c20) {};
	\draw[dspconn] (c20) -- node[midway,above] {\footnotesize $\mathbf{Y}$} (c2);
	
	\tikzset{blue dotted2/.style={draw=black, line width=1pt,
			dash pattern=on 1pt off 4pt on 6pt off 4pt,
			inner sep=4.5mm, rectangle, rounded corners}};
	\node (first dotted box) [blue dotted2, fit = (c2) (c5)] {};
	\node at (first dotted box.south) [below, inner sep=4mm] {\footnotesize Module A};
	
	\tikzset{blue dotted2/.style={draw=black, line width=1pt,
			dash pattern=on 1pt off 4pt on 6pt off 4pt,
			inner sep=4.5mm, rectangle, rounded corners}};
	\node (first dotted box) [blue dotted2, fit = (c1) (c7)] {};
	\node at (first dotted box.south) [below, inner sep=4mm] {\footnotesize Module B};
\end{tikzpicture}}
	\caption{An illustration of the Turbo-CS algorithm in \cite{7912330}.}
	\label{fig:turbocs}
\end{figure}

In the following, we introduce a modified Turbo-CS algorithm for the estimation of $p(\mathbf{G}|\mathbf{Y};\hat{\bm{\tau}}(u))$ in oversampled system, where Turbo-CS  \cite{7912330} is an efficient iterative algorithm when the sensing matrices are partial orthogonal. As illustrated in Fig. \ref{fig:turbocs}, the Turbo-CS algorithm \cite{7912330} consists of two modules. Module A estimates $\mathbf{G}$ based on the measurement $\mathbf{Y}$, where the distribution of the $n_\text{R}$-th column of $\mathbf{G}$, denoted by  $\mathbf{g}_{n_\text{R}}$, is \emph{a priori} approximated as $\mathcal{CN}(\mathbf{g}_{\text{A},n_\text{R}}^{\text{Pri}},v_{\text{A},n_\text{R}}^{\text{Pri}}\mathbf{I}_K)$ from the previous iteration. The posterior distribution of $\mathbf{g}_{n_\text{R}}$ is still Gaussian with mean $\mathbf{g}_{\text{A},n_\text{R}}^{\text{Post}}$ and covariance $v_{\text{A},n_\text{R}}^{\text{Post}}\mathbf{I}_K$ calculated as 
\begin{subequations}\label{equ_postG}
	\begin{align}
			\nonumber&\mathbf{g}^{\text{Post}}_{\text{A},n_\text{R}} =\mathbf{g}^{\text{Pri}}_{\text{A},n_\text{R}}+\left(\mathbf{Z}\mathbf{X}(\hat{\bm{\tau}}(u))\right)^H\left(\mathbf{Z}\mathbf{X}(\hat{\bm{\tau}}(u))(\mathbf{Z}\mathbf{X}(\hat{\bm{\tau}}(u)))^H\right.\\&\hspace{0.3cm}\left.+\sigma_\text{N}^2/v_{\text{A},n_\text{R}}^{\text{Pri}}\mathbf{F}\mathbf{F}^H\right)^{-1}(\mathbf{y}_{n_\text{R}}-\mathbf{Z}\mathbf{X}(\hat{\bm{\tau}}(u))\mathbf{g}^{\text{Pri}}_{\text{A},n_\text{R}})\\\nonumber
			&v^{\text{Post}}_{\text{A},n_\text{R}} = v_{\text{A},n_\text{R}}^{\text{Pri}}-v^{\text{Pri}}_{\text{A},n_\text{R}}\text{Tr}\left\{(\mathbf{Z}\mathbf{X}(\hat{\bm{\tau}}(u)))^H\left(\mathbf{Z}\mathbf{X}(\hat{\bm{\tau}}(u))\right.\right.\\&\hspace{0.3cm}\left.\left.\times (\mathbf{Z}\mathbf{X}(\hat{\bm{\tau}}(u)))^H+\sigma_\text{N}^2/v_{\text{A},n_\text{R}}^{\text{Pri}}\mathbf{F}\mathbf{F}^H\right)^{-1}\mathbf{Z}\mathbf{X}(\hat{\bm{\tau}}(u))\right\}/K.
	\end{align}
\end{subequations}
The matrix inversion involved in \eqref{equ_postG} can be rewritten as
\begin{align}\nonumber
	&\left(\mathbf{Z}\mathbf{X}(\hat{\bm{\tau}}(u))(\mathbf{Z}\mathbf{X}(\hat{\bm{\tau}}(u)))^H+\sigma_\text{N}^2/v_{\text{A},n_\text{R}}^{\text{Pri}}\mathbf{F}\mathbf{F}^H\right)^{-1}\\\nonumber& = \left(\mathbf{F}\left(\mathbf{F}\mathbf{X}(\hat{\bm{\tau}}(u))\mathbf{X}(\hat{\bm{\tau}}(u))^H\mathbf{F}^H+\sigma_\text{N}^2/v_{\text{A},n_\text{R}}^{\text{Pri}}\mathbf{I}_{LM_{\text{OSF}}}\right)\mathbf{F}^H\right)^{-1}\\\nonumber&=\left(\mathbf{F}\left(\mathbf{Q\Sigma Q}^H+\sigma_\text{N}^2/v_{\text{A},n_\text{R}}^{\text{Pri}}\mathbf{QQ}^H\right)\mathbf{F}^H\right)^{-1}\\&=((\mathbf{FQ})^H)^{-1}(\bm{\Sigma }+\sigma_\text{N}^2/v_{\text{A},n_\text{R}}^{\text{Pri}}\mathbf{I}_{LM_{\text{OSF}}})^{-1}(\mathbf{FQ})^{-1},\label{simpli}
\end{align}
where $\mathbf{Z}=\mathbf{FF}^H$ since the transmit and receive filters are matched, $\mathbf{F}\mathbf{X}(\hat{\bm{\tau}}(u))\mathbf{X}(\hat{\bm{\tau}}(u))^H\mathbf{F}^H=\mathbf{Q\Sigma Q}^H$ by using the eigenvalue decomposition, $\mathbf{Q}$ is a unitary matrix and $\mathbf{\Sigma}$ is a diagonal matrix with diagonal elements being the eigenvalues. Through the modifications in \eqref{simpli}, the matrix inversion related to $v_{\text{A},n_\text{R}}^{\text{Pri}}$ is simplified to the inverse of diagonal elements of $(\bm{\Sigma }+\sigma_\text{N}^2/v_{\text{A},n_\text{R}}^{\text{Pri}}\mathbf{I}_{LM_{\text{OSF}}})$, which significantly reduces the computational costs.

The extrinsic mean $\mathbf{g}_{\text{A},n_\text{R}}^{\text{Ext}}$ and variance $v_{\text{A},n_\text{R}}^{\text{Ext}}\mathbf{I}_K$ are given, respectively, as \cite{539767}
\begin{subequations}\label{equ_Gextrinsic}
	\begin{align}
		\mathbf{g}^{\text{Ext}}_{\text{A},n_\text{R}}&=v^{\text{Ext}}_{\text{A},n_\text{R}}(\mathbf{g}^{\text{Post}}_{\text{A},n_\text{R}}/v^{\text{Post}}_{\text{A},n_\text{R}}-\mathbf{g}^{\text{Pri}}_{\text{A},n_\text{R}}/v^{\text{Pri}}_{\text{A},n_\text{R}})\\
		v^{\text{Ext}}_{\text{A},n_\text{R}} &= (1/v^{\text{Post}}_{\text{A},n_\text{R}}-1/v^{\text{Pri}}_{\text{A},n_\text{R}})^{-1}.
	\end{align}
\end{subequations}

Module B estimates $\mathbf{G}$ by combining the prior distribution $\mathbf{g}_{n_\text{R}}\sim\mathcal{CN}(\bm{0}_{K\times 1},\mathbf{\Gamma})$ and the message from Module A, where $\mathbf{\Gamma}\in\mathbb{R}^{K\times K}$ is a diagonal matrix with the $k$-th diagonal element being $\gamma_k$. The \emph{a posteriori} mean $\mathbf{g}^{\text{Post}}_{\text{B},n_\text{R}}$ and variance $v^{\text{Post}}_{\text{B},n_\text{R}}\mathbf{I}_K$ are calculated, respectively, as 
\begin{subequations}\label{equ_postmeanG}
	\begin{align}
	\textsl{g}^{\text{Post}}_{\text{B},k,n_\text{R}}&= \nu_k\frac{\gamma_k}{\gamma_k+v^{\text{Pri}}_{\text{B},n_\text{R}}}\textsl{g}^{\text{Pri}}_{\text{B},k,n_\text{R}}\\\nonumber
	v^{\text{Post}}_{\text{B},n_\text{R}} &=\sum_{k=1}^{K} \frac{\nu_k}{K}\left(\left(\frac{\textsl{g}^{\text{Pri}}_{\text{B},k,n_\text{R}}/v^{\text{Pri}}_{\text{B},n_\text{R}}}{1/\gamma_k+1/v^{\text{Pri}}_{\text{B},n_\text{R}}}\right)^2+\frac{1}{1/\gamma_k+1/v^{\text{Pri}}_{\text{B},n_\text{R}}}\right)\\&\hspace{0.5cm}- \left(\textsl{g}^{\text{Post}}_{\text{B},k,n_\text{R}}\right)^2,
	\end{align}
\end{subequations}
where $\nu_k$ is calculated as 
\begin{equation}\label{pbi}
\resizebox{\columnwidth}{!}{$\nu_k=\left(1+\left(\frac{\rho}{1-\rho}\prod_{n_\text{R}=1}^{N_\text{R}}\frac{\mathcal{CN}\left(\textsl{g}^{\text{Pri}}_{\text{B},k,n_\text{R}};0,\gamma_k+v^{\text{Pri}}_{\text{B},n_\text{R}}\right)}{\mathcal{CN}\left(\textsl{g}^{\text{Pri}}_{\text{B},k,n_\text{R}};0,v^{\text{Pri}}_{\text{B},n_\text{R}}\right)}\right)^{-1}\right)^{-1}$}
\end{equation}
with $\textsl{g}^{\text{Pri}}_{\text{B},k,n_\text{R}}$ being the $k$-th element of $\mathbf{g}^{\text{Pri}}_{\text{B},n_\text{R}}$.

Moreover, since the probability of active users varies from window by window, the update of $\rho$ is calculated as \cite{6556987}
\begin{equation}\label{equ_rhoup}
	\hat{\rho} = \frac{1}{K}\sum_{k=1}^{K}\nu_{k}.
\end{equation}
The detailed operations are outlined
in Algorithm \ref{alg:turbocsmmv}, where $\varepsilon_1$ and $\varepsilon_2$ are predefined small values.
\begin{algorithm}
	\caption{Turbo-CS with colored noise}
	\begin{algorithmic}[1]
		\STATEx \textbf{Input: }$\mathbf{Y},\mathbf{Z},\mathbf{X}(\hat{\bm{\tau}}(u)),\mathbf{F},\mathbf{\Gamma},\sigma_\text{N}^2$
		\STATEx \textbf{Initialization: }$\forall n_\text{R}: $$\mathbf{g}^{\text{Pri}}_{\text{A},n_\text{R}}=\mathbf{0}_{K\times 1},v^{\text{Pri}}_{\text{A},n_\text{R}}=1,\rho(0)=0$
		\FOR{$j_{\text{O}}=1:J_{\text{O}}$}
		\FOR{$j_{\text{I}}=1:J_{\text{I}}$}
		\STATEx \hspace{0.8cm}\% Module A
		\STATE $\forall n_\text{R}: $ calculate the posterior mean $\mathbf{g}^{\text{Post}}_{\text{A},n_\text{R}}(j_{\text{I}})$ and variance $v^{\text{Post}}_{\text{A},n_\text{R}}(j_{\text{I}})$ in \eqref{equ_postG}
		\STATE $\forall n_\text{R}: $ calculate the extrinsic mean $\mathbf{g}^{\text{Ext}}_{\text{A},n_\text{R}}(j_{\text{I}})$ and variance $v^{\text{Ext}}_{\text{A},n_\text{R}}(j_{\text{I}})$ in \eqref{equ_Gextrinsic}
		\STATE $\forall n_\text{R}: $ $ \mathbf{g}^{\text{Pri}}_{\text{B},n_\text{R}}(j_{\text{I}}) = \mathbf{g}^{\text{Ext}}_{\text{A},n_\text{R}}(j_{\text{I}}),v^{\text{Pri}}_{\text{B},n_\text{R}}(j_{\text{I}}) = v^{\text{Ext}}_{\text{A},n_\text{R}}(j_{\text{I}}),$
		\STATEx \hspace{0.8cm}\% Module B
		\STATE $\forall n_\text{R}: $ calculate the posterior mean $\mathbf{g}^{\text{Post}}_{\text{B},n_\text{R}}(j_{\text{I}})$ and variance $v^{\text{Post}}_{\text{B},n_\text{R}}(j_{\text{I}})$ in \eqref{equ_postmeanG}
		\STATE $\forall n_\text{R}: $ calculate the extrinsic mean $\mathbf{g}^{\text{Ext}}_{\text{B},n_\text{R}}(j_{\text{I}})$ and variance $v^{\text{Ext}}_{\text{B},n_\text{R}}(j_{\text{I}})$ similar to \eqref{equ_Gextrinsic}
		\STATE $\forall n_\text{R}: $ $ \mathbf{g}^{\text{Pri}}_{\text{A},n_\text{R}}(j_{\text{I}}+1) = \mathbf{g}^{\text{Ext}}_{\text{B},n_\text{R}}(j_{\text{I}}),v^{\text{Pri}}_{\text{A},n_\text{R}}(j_{\text{I}}+1) = v^{\text{Ext}}_{\text{B},n_\text{R}}(j_{\text{I}})$
		\State \textbf{if}  $\sum_{n_\text{R}}|v^{\text{Post}}_{\text{B},n_\text{R}}(j_{\text{I}})-v^{\text{Post}}_{\text{B},n_\text{R}}(j_{\text{I}}-1)|<\varepsilon_1$ \textbf{stop}
		\ENDFOR
		\STATE Update $\rho(j_{\text{O}})$ in \eqref{equ_rhoup}
		\State \textbf{if} $|\rho(j_{\text{O}})-\rho(j_{\text{O}}-1)|<\varepsilon_2$ \textbf{stop}
		\ENDFOR
		\STATEx \textbf{Output:  }$\forall n_\text{R}: $ $\hat{\mathbf{g}}_{n_\text{R}}(u)=\mathbf{g}^{\text{Post}}_{\text{B},n_\text{R}}$, $\mathbf{V}^{\mathbf{G}}_{n_\text{R}}(u)=v^{\text{Post}}_{\text{B},n_\text{R}}\mathbf{I}_K$
	\end{algorithmic}
	\label{alg:turbocsmmv}
\end{algorithm}

\subsection{Overall Algorithm}\label{subsec_overall}
From the above discussions, we adopt an EM-based framework to compute $\bm{\tau}$ iteratively. With the estimated $\hat{\bm{\tau}}$, the equivalent channel matrix is calculated in \eqref{equ_MMSE} to obtain $\hat{\mathbf{G}}$. Then, the $k$-th user is active if the power of the $k$-th row of $\hat{\mathbf{G}}$ is larger than a predetermined threshold $\eta_\text{th}$, i.e.,
\begin{equation}
	\hat{\alpha}_k=\begin{cases}
		1,\quad \sum_{n_\text{R}=1}^{N_\text{R}}|\hat{\textsl{g}}_{k,n_\text{R}}|^2>\eta_\text{th},\\
		0,\quad \sum_{n_\text{R}=1}^{N_\text{R}}|\hat{\textsl{g}}_{k,n_\text{R}}|^2\leq\eta_\text{th}.
	\end{cases}
\end{equation}
The proposed UAD-DC method is summarized in Algorithm \ref{alg_jdcuad}, where $\varepsilon_3$ is a predefined small value.

In Algorithm \ref{alg:turbocsmmv}, the time delay $\bm{\tau}$ is initially estimated using cross-correlation between the received samples $\mathbf{Y}$ and the $i$-th local preamble sequence ${\mathbf{x}_\text{LO}}_i\in\mathbb{C}^{M_{\text{OSF}}N_{\text{PL}}\times1}$, $i\in \{1,\cdots,N_{\text{P}}\}$ with $N_{\text{P}}$ being the total number of available orthogonal preamble sequences. The $m$-th element of the cross-correlated sequence $\mathbf{r}_{i}$, denoted by $r_{i,m}$, is calculated as
\begin{align}\label{equ_cross}
		\nonumber&r_{i,m}=\sum_{n_\text{R}=1}^{N_\text{R}}\left|\sum_{n=0}^{LM_{\text{OSF}}-m-1}y_{n_\text{R},n+m}{x^*_\text{LO}}_{i,n}\right|, \\&\hspace{2cm}
		m=0,\cdots,LM_{\text{OSF}}+M_{\text{OSF}}N_{\text{PL}}-1
\end{align} 
where $y_{n_\text{R},n+m}$ is the $(n+m)$-th element of $\mathbf{y}_{n_\text{R}}$, ${x_\text{LO}}_{i,n}$ is the $n$-th element of ${\mathbf{x}_\text{LO}}_i$, and $(\cdot)^*$ denotes the operation of complex conjugate. Letting $\eta'_\text{th}$ be a pre-defined threshold, 
the values $\{r_{i,m}\}$ fulfilling $r_{i,m}>\eta'_\text{th}$ are identified, and the corresponding time delays are stored in $\hat{\bm{\tau}}(1)$. 


In general, UAD-DC is composed of the cross-correlation in steps 1-4, Algorithm \ref{alg:turbocsmmv} in step 6 and Algorithm \ref{alg:receiver3} in step 7. The cross-correlation requires $M^2_{\text{OSF}}N_\text{R}N_{\text{P}}N_\text{PL}L$ multiplications. The complexity of Algorithm \ref{alg:turbocsmmv} is dominated by \eqref{equ_postG}, \eqref{equ_Gextrinsic} and \eqref{equ_postmeanG}. Due to matrix multiplications and inversions, \eqref{equ_postG} requires $1.5(LM_{\text{OSF}})^3+6K(LM_{\text{OSF}})^2+(3K+1)(LM_{\text{OSF}})$ multiplications, where the complexity for inverting an $n$-by-$n$ matrix is $0.5n^3$ \cite{6710599}. Equations \eqref{equ_Gextrinsic} and \eqref{equ_postmeanG} require $3K$ and $(2N_\text{R}^2+17)K$ multiplications, respectively. In total, Algorithm \ref{alg:turbocsmmv} requires $1.5(LM_{\text{OSF}})^3+5K(LM_{\text{OSF}})^2+J_{\text{O}}J_{\text{I}}N_\text{R}(K(LM_{\text{OSF}})^2+(1+3K)(LM_{\text{OSF}})+(2N_\text{R}^2+20)K)$ multiplications. The complexity of Algorithm \ref{alg:receiver3} is dominated by \eqref{equ_optmax}, which requires $2K(LM_{\text{OSF}})^2+K^2(LM_{\text{OSF}})+(LM_{\text{OSF}}K+K^3+K^2)N_\text{R}$ multiplications. In total, Algorithm \ref{alg:receiver3} requires $2K\epsilon(2K(LM_{\text{OSF}})^2+K^2(LM_{\text{OSF}})+(KLM_{\text{OSF}}+K^3+K^2)N_\text{R})$ multiplications. To summarize, UAD-DC requires $M^2_{\text{OSF}}N_\text{R}N_{\text{P}}N_\text{PL}L+U(1.5(LM_{\text{OSF}})^3+5K(LM_{\text{OSF}})^2+J_{\text{O}}J_{\text{I}}N_\text{R}(K(LM_{\text{OSF}})^2+(1+3K)(LM_{\text{OSF}})+(2N_\text{R}^2+20)K)+2K\epsilon(2K(LM_{\text{OSF}})^2+K^2(LM_{\text{OSF}})+(KLM_{\text{OSF}}+K^3+K^2)N_\text{R}))$ multiplications.
\begin{algorithm}[!htbp] 
	\caption{User activity detection with delay-calibration (UAD-DC)}
	\begin{algorithmic}[1] 
		\STATEx \textbf{Input: }$\mathbf{Y}$, $\mathbf{Z}$, $\mathbf{F}$, $\mathbf{\Gamma}$, $\sigma_\text{N}^2$
		\STATEx \% Initialization
		\FOR{$i=1:N_{\text{P}}$} 
		\STATE Calculate the cross-correlation in \eqref{equ_cross} 
		\STATE Find and store time delays based on $\eta'_{\text{th}}$
		\ENDFOR
		\STATEx \% Main part
		\FOR{$u=1:U$} 
		\STATE $\forall n_\text{R}: $ use Algorithm \ref{alg:turbocsmmv} to obtain $\hat{\mathbf{g}}_{n_\text{R}}(u)$ and $\mathbf{V}^{\mathbf{G}}_{n_\text{R}}(u)$ 
		\STATE Use Algorithm \ref{alg:receiver3} to obtain $\hat{\bm{\tau}}(u+1)$
		\State \textbf{if} $|\hat{\bm{\tau}}(u+1)-\hat{\bm{\tau}}(u)|<\varepsilon_3$ \textbf{stop}
		\ENDFOR
		\STATEx \textbf{Output: }$\hat{\bm{\tau}}, \hat{\mathbf{G}}$
	\end{algorithmic} 
	\label{alg_jdcuad}
\end{algorithm} 

\section{Bayesian Cram\'er-Rao Bound}
As illustrated in Section \ref{sec2em}, the calculation of E-step in \eqref{equ_expec} is based on the posterior density $p(\mathbf{G}|\mathbf{Y};\hat{\bm{\tau}}(u))$ with mean and variance obtained from Algorithm \ref{alg:turbocsmmv}. From \eqref{equ_MMSE}, the mean of $p(\mathbf{G}|\mathbf{Y};\hat{\bm{\tau}}(u))$ is exactly the MMSE estimate of $\mathbf{G}$ in the $u$-th EM iteration, namely $\hat{\mathbf{G}}(u)$. In this section, we develop a mean square error lower bound for $\mathbf{G}$ based on BCRB to verify the efficacy of Algorithm \ref{alg:turbocsmmv}. Moreover, the derived BCRB can be treated as an objective function for the  optimization of pulse shaping filter to improve the user activity detection performance in oversampled system, which will be discussed in Section \ref{optp}.

The non-sparse signal model in \eqref{equ_sys1} is written as 
\begin{equation}\label{equ_bcrb1}
		\mathbf{Y}=
		\mathbf{Z}\mathbf{X}_{\text{real}}(\bm{\tau}_{\text{real}})\mathbf{G}_{\text{real}} + \mathbf{FW},
\end{equation}
where $\mathbf{X}_{\text{real}}(\bm{\tau}_{\text{real}})\in\mathbb{C}^{LM_{\text{OSF}}\times K_{\text{real}}}$ is constructed by preamble sequences transmitted from $K_{\text{real}}$ active users with time delays $\bm{\tau}_{\text{real}}\in\mathbb{R}^{K_{\text{real}}\times 1}$, and $\mathbf{G}_{\text{real}}\in\mathbb{C}^{K_{\text{real}}\times N_\text{R}}$ is the corresponding channel matrix. By stacking the columns of \eqref{equ_bcrb1} sequentially on top of one another, we obtain
\begin{equation}\label{equ_bcrb2}
	\mathbf{y}=
	(\mathbf{I}_{N_\text{R}}\otimes\mathbf{Z}\mathbf{X}_{\text{real}}(\bm{\tau}_{\text{real}}))\mathbf{g}_{\text{real}} + (\mathbf{I}_{N_\text{R}}\otimes\mathbf{F})\mathbf{w},
\end{equation}
where $\mathbf{g}_{\text{real}}\in\mathbb{C}^{K_{\text{real}}N_\text{R}\times1}$ and $\mathbf{w}\in\mathbb{C}^{LM_{\text{OSF}}N_\text{R}\times1}$ are vector forms of $\mathbf{G}_{\text{real}}$ and $\mathbf{W}$, respectively. The Bayesian information matrix (BIM) \cite{Trees} is defined as
\begin{equation}
	\begin{aligned}
		\mathbf{J}_{\mathbf{y}}(\mathbf{g}_{\text{real}})=\mathbf{J}^\mathrm{D}_{\mathbf{y}}(\mathbf{g}_{\text{real}})+\mathbf{J}^\mathrm{P}_{\mathbf{y}}(\mathbf{g}_{\text{real}}).
	\end{aligned}
	\label{equ_bayesian}
\end{equation}
Specifically, the $(i,j)$-th elements of $\mathbf{J}^\mathrm{D}_{\mathbf{y}}(\mathbf{g}_{\text{real}})$ and $\mathbf{J}^\mathrm{P}_{\mathbf{y}}(\mathbf{g}_{\text{real}})$ are calculated, respectively, as 
\begin{subequations}
	\begin{align}
		[\mathbf{J}^\mathrm{D}_{\mathbf{y}}(\mathbf{g}_{\text{real}})]_{i,j}&=E_{\mathbf{y}|\mathbf{g}_{\text{real}}}\left\{\frac{\partial \ln p(\mathbf{y}|\mathbf{g}_{\text{real}})}{\partial \textsl{g}_{\text{real},i}}\frac{\partial \ln p(\mathbf{y}|\mathbf{g}_{\text{real}})}{\partial \textsl{g}_{\text{real},j}}\right\}
		\label{equ_bayesianjd}\\
		[\mathbf{J}^\mathrm{P}_{\mathbf{y}}(\mathbf{g}_{\text{real}})]_{i,j} &= E_{\mathbf{g}_{\text{real}}}\left\{\frac{\partial \ln p(\mathbf{g}_{\text{real}})}{\partial \textsl{g}_{\text{real},i}}\frac{\partial \ln p(\mathbf{g}_{\text{real}})}{\partial \textsl{g}_{\text{real},j}}\right\}
		\label{equ_bayesianjp}
	\end{align}
\end{subequations}
with $\textsl{g}_{\text{real},i}$ and $\textsl{g}_{\text{real},j}$ being the $i$-th and $j$-th elements of $\mathbf{g}_{\text{real}}$, respectively. Since $p(\mathbf{y}|\mathbf{g}_{\text{real}})\sim\mathcal{CN}((\mathbf{I}_{N_\text{R}}\otimes\mathbf{Z}\mathbf{X}_{\text{real}}(\bm{\tau}_{\text{real}}))\mathbf{g}_{\text{real}}, \sigma_\text{N}^2(\mathbf{I}_{N_\text{R}}\otimes\mathbf{FF}^H))$, \eqref{equ_bayesianjd} is calculated as
\begin{align}\label{equjd}
	\nonumber[\mathbf{J}^\mathrm{D}_{\mathbf{y}}(\mathbf{g}_{\text{real}})]_{i,j}=&\mathbf{e}^H_i(\mathbf{I}_{N_\text{R}}\otimes\mathbf{ZX}_{\text{real}}(\bm{\tau}_{\text{real}}))^H(\sigma_\text{N}^2\mathbf{I}_{N_\text{R}}\otimes\mathbf{FF}^H)^{-1}\\\nonumber&\times(\mathbf{I}_{N_\text{R}}\otimes\mathbf{ZX}_{\text{real}}(\bm{\tau}_{\text{real}}))\mathbf{e}_j\\&\hspace{-2.4cm}=\frac{1}{\sigma_\text{N}^2}\mathbf{e}^H_i(\mathbf{I}_{N_\text{R}}\otimes\mathbf{X}_{\text{real}}^H(\bm{\tau}_{\text{real}})\mathbf{Z}(\mathbf{FF}^H)^{-1}\mathbf{ZX}_{\text{real}}(\bm{\tau}_{\text{real}}))\mathbf{e}_j,
\end{align}
where the unit vector $\mathbf{e}_i$ and $\mathbf{e}_j$ are all-zero vectors except that the $i$-th and $j$-th element are ones. The prior density $p(\mathbf{g}_{\text{real}})$ follows a complex Gaussian distribution $\mathcal{CN}(\mathbf{0}_{K_{\text{real}}N_\text{R}\times 1},\mathbf{I}_{N_\text{R}}\otimes\mathbf{\Gamma}_{\text{real}})$, where $\mathbf{\Gamma}_{\text{real}}\in\mathbb{R}^{K_{\text{real}}\times K_{\text{real}}}$ is a diagonal matrix similar to $\mathbf{\Gamma}$. Then, $\ln p(\mathbf{g}_{\text{real}})$ yields
\begin{align}
	\nonumber\ln p(\mathbf{g}_{\text{real}}) =& -\frac{K_{\text{real}}N_\text{R}}{2}\ln(\pi)-\frac{N_\text{R}}{2}\sum_{k=1}^{K_{\text{real}}}\ln\gamma_k\\&-\mathbf{g}_{\text{real}}^H(\mathbf{I}_{N_\text{R}}\otimes\mathbf{\Gamma}_{\text{real}})^{-1}\mathbf{g}_{\text{real}}.
\end{align}
By the definition in \eqref{equ_bayesianjp}, we have
\begin{equation}\label{equjp}
		[\mathbf{J}^\mathrm{P}_{\mathbf{y}}(\mathbf{g}_{\text{real}})]_{i,j} = 2(\mathbf{I}_{N_\text{R}}\otimes\mathbf{\Gamma}_{\text{real}}^{-1}).
\end{equation}
Inserting \eqref{equjd} and \eqref{equjp} into \eqref{equ_bayesian}, we obtain
\begin{align}
	\nonumber\mathbf{J}_{\mathbf{y}}(\mathbf{g}_{\text{real}})&=\frac{1}{\sigma_\text{N}^2}(\mathbf{I}_{N_\text{R}}\otimes\mathbf{X}_{\text{real}}^H(\bm{\tau}_{\text{real}})\mathbf{Z}(\mathbf{FF}^H)^{-1}\mathbf{ZX}_{\text{real}}(\bm{\tau}_{\text{real}}))\\&\nonumber\hspace{0.4cm}+2(\mathbf{I}_{N_\text{R}}\otimes\mathbf{\Gamma}_{\text{real}}^{-1})\\&\hspace{-1.3cm}=\mathbf{I}_{N_\text{R}}\otimes(\frac{1}{\sigma_\text{N}^2}\mathbf{X}_{\text{real}}^H(\bm{\tau}_{\text{real}})\mathbf{Z}(\mathbf{FF}^H)^{-1}\mathbf{ZX}_{\text{real}}(\bm{\tau}_{\text{real}})+2\mathbf{\Gamma}_{\text{real}}^{-1}).
\end{align}
The BCRB is given by
\begin{align}
	\nonumber&\text{Var}\{\mathbf{g}_{\text{real}}\}=\text{Tr}\{\mathbf{J}^{-1}_{\mathbf{y}}(\mathbf{g}_{\text{real}})\}\\\nonumber&=N_\text{R}\text{Tr}\{(\frac{1}{\sigma_\text{N}^2}\mathbf{X}_{\text{real}}^H(\bm{\tau}_{\text{real}})\mathbf{Z}(\mathbf{FF}^H)^{-1}\mathbf{ZX}_{\text{real}}(\bm{\tau}_{\text{real}})+2\mathbf{\Gamma}_{\text{real}}^{-1})^{-1}\}\\\label{equ_bcrbfin}&=N_\text{R}\text{Tr}\{(\frac{1}{\sigma_\text{N}^2}\mathbf{X}_{\text{real}}^H(\bm{\tau}_{\text{real}})\mathbf{ZX}_{\text{real}}(\bm{\tau}_{\text{real}})+2\mathbf{\Gamma}_{\text{real}}^{-1})^{-1}\},
\end{align}
where the last equality comes from $\mathbf{Z}=\mathbf{FF}^H$, since the transmit and receive filters are matched.
The variance of the estimator $\hat{\mathbf{g}}$ is lower bounded by \cite{Trees}
\begin{equation}
	\text{Var}\{\hat{\mathbf{g}}\}\geq\text{Var}\{\mathbf{g}_{\text{real}}\},
\end{equation}
where $\hat{\mathbf{g}}$ is a vector form of $\hat{\mathbf{G}}$ similar to $\mathbf{g}_{\text{real}}$. 

\section{Pulse Shaping Filter Optimization}\label{optp}
To improve the performance of user activity detection, we propose a waveform optimization of the pulse shaping filter based on the derived BCRB.
\subsection{Preliminary Discussions} \label{chap_trreceiver}
In wireless communication systems, pulse shaping filters are used to transfer digital data through a band-limited channel by converting it to an equivalent modulated analog signal. Basically, there are two design criteria for the pulse shaping filters, namely the bandwidth in frequency domain and the ISI in time domain. Good pulse shaping filters can limit the bandwidth of the transmission and filter the pulses with low ISI. Some widely used pulse shaping filters are:
\begin{itemize}
	\item RRC filter \cite{wcp}: It is frequently used in a digital communication system to perform pulse shaping and matched filtering. This type of filter helps to avoid the ISI. The combined response of two such filters is that of the raised-cosine (RC) filter, which is ideally band-limited. The RRC filter is characterized by two values, namely the roll-off factor $\beta$ and the symbol duration $T_\text{S}$.
	\item Gaussian filter \cite{wcp}: The impulse response of a Gaussian filter is a Gaussian function. A Gaussian filter has the advantage that its Fourier transform is also a Gaussian distribution centered around the zero frequency. One can then control the effectiveness of the low-pass nature of the filter by adjusting its variance $\sigma_\text{G}^2$. 
\end{itemize}
Several examples are shown in Fig. \ref{fig_allfilter}, where $z(t)$ represents the convolution of the pulse shaping and its corresponding matched filter. 
From Fig. \ref{fig_timed}, the Gaussian filter with $\sigma_\text{G}^2=0.1$ causes the lowest ISI when it is sampled at the oversampling rate. This leads to the best BCRB performance, as shown in Fig. \ref{fig_filterperformance}, where the simulation parameters are outlined in Section \ref{numer}. However, the Gaussian filter with $\sigma_\text{G}^2=0.1$ consumes the largest bandwidth from Fig. \ref{fig_frequencyd}, which is not spectrally efficient. 

Based on the above discussions, it is important to design a filter that is suitable for the oversampled system. From \eqref{equ_bcrbfin}, the BCRB is related to the matrix $\mathbf{Z}$, which is constructed by $z(t)$ at different time instants according to \eqref{eq_zform}. In the following, we present a pulse shaping filter optimization algorithm for optimizing $z(t)$ with the objective of minimizing the BCRB under the constraint of limited spectral bandwidth.



\begin{figure}[!htbp]
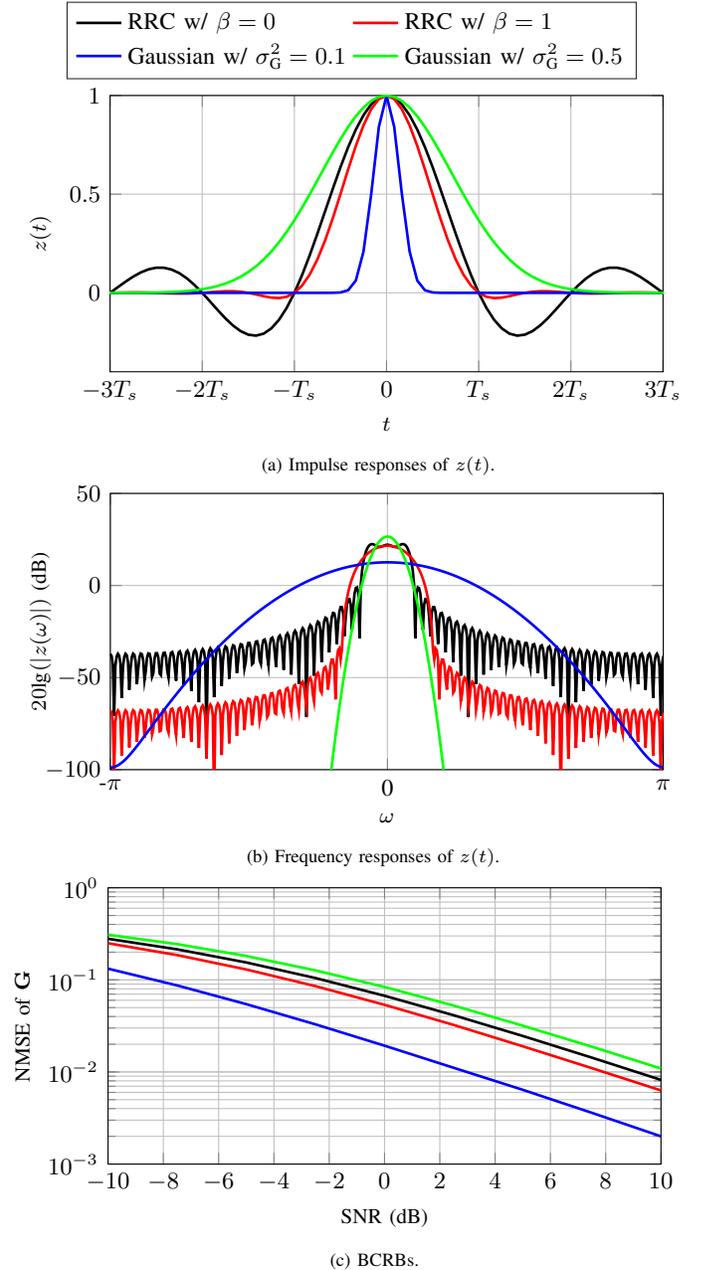

	\centering
	\begin{minipage}{0.83\columnwidth}
	\centering
	\begin{tikzpicture} 
	\begin{axis}[%
		hide axis,
		xmin=10,
		xmax=40,
		ymin=0,
		ymax=0.4,
		ylabel style={font=\scriptsize},
		xlabel style={font=\scriptsize},
		legend columns=2,
		legend style={draw=white!15!black,legend cell align=left}
		]
		
		\addlegendimage{black,line width=1.0pt}
		\addlegendentry{RRC w/ $\beta=0$};
		
		\addlegendimage{red,line width=1.0pt}
		\addlegendentry{RRC w/ $\beta=1$};
		
		\addlegendimage{blue,line width=1.0pt}
		\addlegendentry{Gaussian w/ $\sigma_\text{G}^2=0.1$};
		
		\addlegendimage{green,line width=1.0pt}
		\addlegendentry{Gaussian w/ $\sigma_\text{G}^2=0.5$};
		
	\end{axis}
\end{tikzpicture}
	\subcaptionbox{Impulse responses of $z(t)$.\label{fig_timed}}{\hspace{-0.6cm}\input{timedomain.tex}}
	\subcaptionbox{Frequency responses of $z(t)$.\label{fig_frequencyd}}{\hspace{-0.6cm}\input{frequencydomain.tex}}
	\subcaptionbox{BCRBs.\label{fig_filterperformance}}{\hspace{-0.8cm}\begin{tikzpicture}
	
\begin{axis}[%
	width=\columnwidth,
	height=.5\columnwidth,
	at={(0in,0in)},
	scale only axis,
	xmin=-10,
	xmax=10,
	xlabel={SNR (dB)},
	xtick={-10,-8,-6,-4,-2,0,2,4,6,8,10},
	xmajorgrids,
	ylabel style={font=\footnotesize},
	xlabel style={font=\footnotesize},
	ymode=log,
	ymin=0.001,
	ymax=1,
	yminorticks=true,
	ylabel={NMSE of $\mathbf{G}$},
	axis background/.style={fill=white},
	ymajorgrids,
	yminorgrids,
	]
	
	\addplot [color=red, solid,line width=1pt]
	table[row sep=crcr]{%
		-10	  0.2498\\
		-7.5  0.1858\\
		-5    0.1297\\
		-2.5  0.0855\\
		0     0.0536\\
		2.5   0.0324\\
		5     0.0190\\
		7.5   0.0110\\
		10    0.0063\\
	};

	\addplot [color=black, solid,line width=1pt]
	table[row sep=crcr]{%
		-10	  0.2795\\
		-7.5  0.2146\\
		-5    0.1545\\
		-2.5  0.1047\\
		0     0.0672\\
		2.5   0.0413\\
		5     0.0246\\
		7.5   0.0143\\
		10    0.0082\\
	};

	\addplot [color=blue, solid,line width=1pt]
	table[row sep=crcr]{%
		-10	  0.1323\\
		-7.5  0.0870\\
		-5    0.0545\\
		-2.5  0.0329\\
		0     0.0193\\
		2.5   0.0111\\
		5     0.0064\\
		7.5   0.0036\\
		10    0.0020\\
	};

	\addplot [color=green, solid,line width=1pt]
	table[row sep=crcr]{%
		-10	  0.3087\\
		-7.5  0.2445\\
		-5    0.1816\\
		-2.5  0.1268\\
		0     0.0837\\
		2.5   0.0526\\
		5     0.0318\\
		7.5   0.0188\\
		10    0.0109\\
	};
\end{axis}
\end{tikzpicture}%
}
	\end{minipage}
	\caption{Performance comparisons of different types of pulse shaping filters, where $z(t)$ represents the convolution of the pulse shaping and its corresponding matched filter.}\label{fig_allfilter}
\end{figure}

\subsection{Waveform Optimization}
The goal of the waveform optimization is to minimize the BCRB derived in \eqref{equ_bcrbfin} under the pre-defined spectral bandwidth. This yields the following optimization problem:
\begin{equation}\label{equ_opt}
	\begin{split}
		&\underset{\mathbf{z}}{\text{minimize}} \qquad \text{Var}\{\mathbf{g}_{\text{real}}\}\\
		&\text{subject to}\qquad |\mathscr{F}(\mathbf{z})|\leq \mathbf{b}_\text{up},
	\end{split}
\end{equation}
where $\mathscr{F}(\cdot)$ denotes the fast Fourier transform, and $\mathbf{b}_\text{up}$ represents the upper bound of the pre-defined spectral mask. 
From \eqref{eq_zform}, the matrix $\mathbf{Z}$ can be represented by the vector  $\mathbf{z}=[z[0],z[\frac{1}{M_{\text{OSF}}}],\cdots,z[3]]$ in the following form:
\begin{equation}
	\begin{split}
		\mathbf{Z} =& z\left[0\right]\mathbf{I}_{LM_{\text{OSF}}}+z\left[\frac{1}{M_{\text{OSF}}}\right]\mathbf{C}_1+z\left[\frac{2}{M_{\text{OSF}}}\right]\mathbf{C}_2\\&+\cdots+z\left[3\right]\mathbf{C}_{3M_{\text{OSF}}},
	\end{split}
	\label{eq_zzform}
\end{equation}
where $\mathbf{C}_i\in\mathbb{R}^{LM_{\text{OSF}}\times LM_{\text{OSF}}}$, $i\in [1,3M_{\text{OSF}}]$, is a constant Toeplitz matrix only containing zeros and ones. The first row and the first column of $\mathbf{C}_i$ share the same form as
\begin{equation}
	[\quad\underbrace{0\cdots0}_{i}\quad1\hspace{0.5em}\underbrace{0\cdots0}_{LM_{\text{OSF}}-i-1}].
\end{equation}
Inserting \eqref{eq_zzform} into \eqref{equ_bcrbfin}, we have
\begin{equation}\label{equ_newopt}
	\begin{split}
		&\text{Var}\{\mathbf{g}_\text{real}\}=N_\text{R}\text{Tr}\{\mathbf{A}^{-1}(\mathbf{z})\},
	\end{split}
\end{equation}
where $\mathbf{A}(\mathbf{z})=\frac{1}{\sigma_\text{N}^2}(z[0]\mathbf{X}_{\text{real}}^H(\bm{\tau}_{\text{real}})\mathbf{X}_{\text{real}}(\bm{\tau}_{\text{real}})+z\left[\frac{1}{M_{\text{OSF}}}\right]\mathbf{X}_{\text{real}}^H(\bm{\tau}_{\text{real}})\mathbf{C}_1\mathbf{X}_{\text{real}}(\bm{\tau}_{\text{real}})+\cdots+z[3]\mathbf{X}_{\text{real}}^H(\bm{\tau}_{\text{real}})\mathbf{C}_{3M_{\text{OSF}}}\mathbf{X}_{\text{real}}(\bm{\tau}_{\text{real}}))+2\mathbf{\Gamma}_{\text{real}}^{-1}$ 
is symmetric positive definite. The optimization problem \eqref{equ_opt} is then converted to
\begin{equation}\label{equ_convexproblem}
\begin{split}
	&\underset{\mathbf{z}}{\text{minimize}} \qquad N_\text{R}\text{Tr}\left\{\mathbf{A}^{-1}(\mathbf{z})\right\}\\
	&\text{subject to}\qquad |\mathscr{F}(\mathbf{z})|\leq \mathbf{b}_\text{up}.
\end{split}
\end{equation}
\begin{prop}
	The problem in \eqref{equ_convexproblem} is convex.
\end{prop}
\begin{proof}
	We start with the objective function in \eqref{equ_convexproblem}. By definition, $\mathbf{A}(\mathbf{z})$ is a positive definite matrix. Thus, from \cite[pp.~116]{boyd_vandenberghe_2004}, $\text{Tr}\{\mathbf{A}^{-1}(\mathbf{z})\}$ is convex with respect to $\mathbf{A}(\mathbf{z})$. We further note that 
	$\mathbf{A}(\mathbf{z})$ is a linear transform of $\mathbf{z}$. Since composition with a linear transformation preserves convexity, we obtain that the objective function in \eqref{equ_convexproblem} is convex. Furthermore, the constraint function is also convex due to the linear fast Fourier transformation, denoted by $\mathscr{F}(\cdot)$, and the convex absolute value, denoted by $|\cdot|$. Therefore, \eqref{equ_convexproblem} is a convex optimization problem.
\end{proof}
The optimization problem \eqref{equ_convexproblem} can be solved by the CVX toolbox \cite{cvx}.
After obtaining $z(t)$, the pulse shaping filter $q(t)$ can be obtained via spectral factorization by the Fej$\acute{e}$r-Riesz theorem. The exact factorization steps can be found in Section 4 of \cite{574313}, and the corresponding MATLAB code is available in \cite{cvx}. In practice, $\mathbf{A}(\mathbf{z})$ can be obtained with averaged $\bm{\tau}_{\text{real}}$ and $\mathbf{\Gamma}_{\text{real}}$ in any given coherence interval.

\section{Numerical Results}\label{numer}
In this section, the setup is as follows. There are 1000 single-antenna users located in the cell, and only 300 of them become active. In the cell, the BS is equipped with $N_\text{R}=32$ antennas, and there are $N_{\text{P}}=64$ available orthogonal preambles in total. Each preamble has length $N_{\text{PL}}=139$ \cite{3gpp}. To align with existing communication systems, we adopt the ZC sequences \cite{4119357} as the preamble sequences due to their good auto-correlation properties and low peak-to-average power ratios. A ZC sequence is generated by two key parameters, namely the root index $v\in[1,,N_{\text{PL}}-1]$ and the length of the sequence $N_{\text{PL}}$. Given these two parameters, $x_{k,n}$ in \eqref{equ_x_S} is calculated as $\exp(-j\pi vn(n+1)/N_{\text{PL}})$ \cite{3gpp}.
Moreover, the RRC filter with $\beta=0.4$ is used for the pulse shaping filter if no other filter is specified. 
Each $\gamma_k$ in dB is randomly generated from a uniform distribution over $[-128.1,-118.1]$. The signal-to-noise ratio (SNR) is defined as $10\lg{\frac{\text{Tr}\{(\mathbf{Z}\mathbf{X}(\bm{\tau})\mathbf{G})^H(\mathbf{Z}\mathbf{X}(\bm{\tau})\mathbf{G})\}}{LM_{\text{OSF}}N_\text{R}\sigma_\text{N}^2}}$. To evaluate the performance of the proposed UAD-DC algorithm, we use the probability of misdetection and the NMSE of $\mathbf{G}$ as the performance metrics.
Note that the probability of misdetection is defined as the average of non-detected active users over the number of active users, and the probability of false alarm is defined as the average of inactive users detected as active over the number of detected users.  

To evaluate the effectiveness of Algorithm \ref{alg:turbocsmmv}, we compare the estimation performance of Algorithm \ref{alg:turbocsmmv} and the BCRB derived in \eqref{equ_bcrbfin} under fixed time delay $\bm{\tau}(u)$ in Fig. \ref{fig_bcrb}. Simulation results show that Algorithm \ref{alg:turbocsmmv} can achieve a performance close to the BCRB. In addition, we compare the sparse recovery performance of Algorithm \ref{alg:turbocsmmv} with that of BG-GAMP \cite{6556987} and VAMP \cite{8713501}. From Fig. \ref{fig_bcrb}, Algorithm \ref{alg:turbocsmmv} outperforms BG-GAMP and VAMP, since it takes the correlation of noise samples into consideration, as shown in \eqref{equ_postG}.
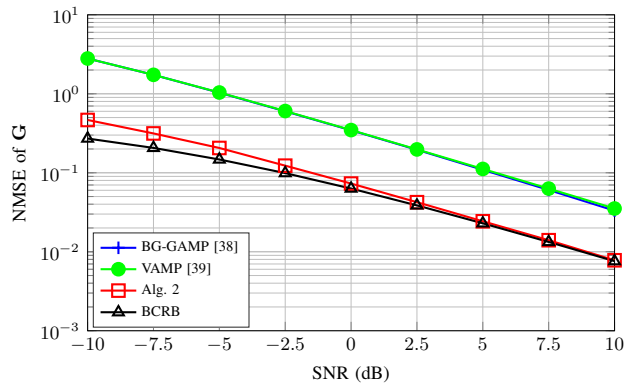
\begin{figure}[!htbp]
	\centering
	\resizebox{0.95\columnwidth}{!}{\begin{tikzpicture}
	
	\begin{axis}[%
		width=\columnwidth,
		height=.6\columnwidth,
		at={(0.758in,0.603in)},
		scale only axis,
		xmin=-10,
		xmax=10,
		ylabel style={font=\small},
		xlabel style={font=\small},
		xlabel={SNR (dB)},
		xtick=data,
		xmajorgrids,
		ymode=log,
		ymin=0.001,
		ymax=10,
		yminorticks=true,
		ylabel={NMSE of $\mathbf{G}$},
		axis background/.style={fill=white},
		ymajorgrids,
		yminorgrids,
		legend style={at={(0.01,0.01)},anchor=south west,legend cell align=left,align=left,draw=white!15!black,font=\scriptsize}
		]
		
		\addplot [color=blue, solid,line width=1pt,mark=+,mark options={solid},mark size=3pt]
		table[row sep=crcr]{
			-10	  2.79683145844058\\
			-7.5  1.73729727967008\\
			-5    1.03228670827045\\
			-2.5  0.599746086087984\\
			0     0.344992190177602\\
			2.5   0.195529834771059\\
			5     0.109013510826974\\
			7.5   0.0605702431793679\\
			10    0.0331313863312968\\
		};
		\addlegendentry{BG-GAMP \cite{6556987}}
		
		\addplot [color=green, solid,line width=1pt,mark=*,mark options={solid},mark size=3pt]
		table[row sep=crcr]{
			-10	  2.79683187753218\\
			-7.5  1.73729857967737\\
			-5    1.03785848036730\\
			-2.5  0.604822737811919\\
			0     0.347200081839582\\
			2.5   0.197553849356784\\
			5     0.111833469783432\\
			7.5   0.0630453090723292\\
			10    0.0352771855090689\\
		};
		\addlegendentry{VAMP \cite{8713501}}
		
		
		\addplot [color=red, solid,line width=1pt,mark=square,mark options={solid},mark size=3pt]
		table[row sep=crcr]{%
			-10	  0.4660\\
			-7.5  0.3147\\
			-5    0.2060\\
			-2.5  0.1230\\
			0     0.0730\\
			2.5   0.0424\\
			5     0.0243\\
			7.5   0.0140\\
			10    0.0078\\
		};
		\addlegendentry{Alg. 2}
		
		\addplot [color=black, solid,line width=1pt,mark=triangle,mark options={solid},mark size=3pt]
		table[row sep=crcr]{%
			-10	  0.2712\\
			-7.5  0.2064\\
			-5    0.1473\\
			-2.5  0.0990\\
			0     0.0631\\
			2.5   0.0386\\
			5     0.0229\\
			7.5   0.0133\\
			10    0.0076\\
		};
		\addlegendentry{BCRB}
	\end{axis}
\end{tikzpicture}}
	\caption{The NMSE of $\mathbf{G}$ by using different message passing-based sparse recovery algorithms.}\label{fig_bcrb}
\end{figure}
Moreover, the window width $L$ must be chosen larger than the preamble length $N_{\text{PL}}$, i.e., $L>N_{\text{PL}}$, to ensure the presence of Type \uppercase\expandafter{\romannumeral1} preambles. To detect more Type \uppercase\expandafter{\romannumeral1} preambles within one window, $L$ should be large enough. However, from the complexity analysis in Section \ref{subsec_overall}, $L$ can not be very large as the computational complexity is proportional to $L^2$. In our simulations, $L$ is chosen to be $187$ for balancing the detection efficiency and the computational cost.
	

Fig. \ref{fig_em} shows the probability of misdetection against the number of EM iterations $u$ in UAD-DC where SNR = 10 dB. The decreasing curves show that the probability of misdetection improves gradually with an increasing number of EM iterations, which confirms the effectiveness of the proposed algorithm. From Fig. \ref{fig_em}, it can be seen that the performance converges only after two EM iterations. This indicates that most collided users can be successfully detected in the first EM iteration, and only a few collided users require to be detected in the remaining EM iterations, thereby demonstrating the fast convergence of the proposed algorithm.
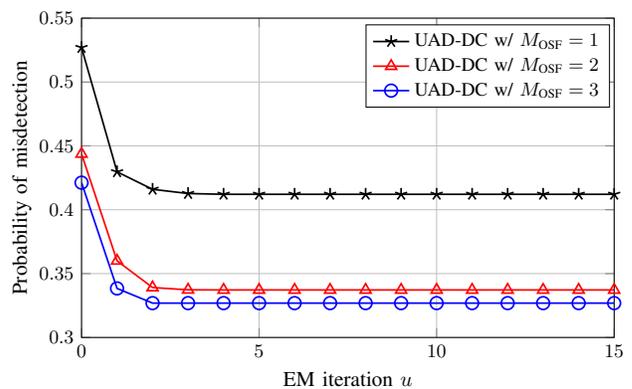
\begin{figure}[!htbp]
	\centering
	\resizebox{0.95\columnwidth}{!}{\begin{tikzpicture}
	
	\begin{axis}[%
		width=\columnwidth,
		height=.6\columnwidth,
		at={(0.758in,0.603in)},
		scale only axis,
		xmin=0,
		xmax=15,
		xlabel={EM iteration $u$},
		ylabel style={font=\normalsize},
		xlabel style={font=\normalsize},
		xtick={0,5,10,15},
		xmajorgrids,
		ymin=0.3,
		ymax=0.55,
		yminorticks=true,
		ylabel={Probability of misdetection},
		ylabel style={at={(axis description cs:-0.08,0.5)}},
		ymajorgrids,
		yminorgrids,
		axis background/.style={fill=white},
		legend style={at={(axis cs:8,0.48)},anchor=south west}
		]
		
		\addplot [color=black,solid,line width=0.8pt,mark=star,mark options={solid},mark size=3pt]
		table[row sep=crcr]{%
			0  0.5270\\
			1  0.4297\\
			2  0.4160\\
			3  0.4127\\
			4  0.4121\\
			5  0.4121\\
			6  0.4121\\
			7  0.4121\\
			8  0.4121\\
			9  0.4121\\
			10 0.4121\\
			11 0.4121\\
			12 0.4121\\
			13 0.4121\\
			14 0.4121\\
			15 0.4121\\
		};	
		\addlegendentry{UAD-DC w/ $M_{\text{OSF}}=1$}
		
		\addplot [color=red,solid,line width=0.8pt,mark=triangle,mark options={solid},mark size=3pt]
		table[row sep=crcr]{%
			0  0.4437\\
			1  0.36\\  
			2  0.3391\\
			3  0.3373\\ 
			4  0.3372\\
			5  0.3372\\
			6  0.3372\\
			7  0.3372\\
			8  0.3372\\
			9  0.3372\\
			10 0.3372\\
			11 0.3372\\
			12 0.3372\\
			13 0.3372\\
			14 0.3372\\
			15 0.3372\\
		};	
		\addlegendentry{UAD-DC w/ $M_{\text{OSF}}=2$}
		
		\addplot [color=blue,solid,line width=0.8pt,mark=o,mark options={solid},mark size=3pt]
		table[row sep=crcr]{%
			0  0.4212\\
			1  0.3384\\
			2  0.3269\\
			3  0.3269\\
			4  0.3269\\
			5  0.3269\\
			6  0.3269\\
			7  0.3269\\
			8  0.3269\\
			9  0.3269\\
			10  0.3269\\
			11  0.3269\\
			12  0.3269\\
			13  0.3269\\
			14  0.3269\\
			15  0.3269\\
		};	
		\addlegendentry{UAD-DC w/ $M_{\text{OSF}}=3$}
	\end{axis}
	
\end{tikzpicture}}
	\caption{Detection performance against the number of EM iterations $u$, where the probability of false alarm is lower than $10^{-3}$ and SNR = 10 dB.}\label{fig_em}
\end{figure}

Fig. \ref{fig_peralg} compares the detection performance of UAD-DC and conventional cross-correlation-based RA. Compared to UAD-DC, conventional RA requires only $M^2_{\text{OSF}}N_\text{R}N_{\text{P}}N_\text{PL}L$ multiplications. Table \ref{tab_complexity} shows numerical complexity comparisons based on the described parameter settings. From Fig. \ref{fig_peralg}, it can be seen that oversampling can improve the detection performance for both algorithms, since accurate estimation of time delays further helps to detect more active users. Unlike the conventional RA algorithm, UAD-DC estimates the time delays using the EM framework iteratively, resulting in more accurate detection of active users. For example, when SNR = 0 dB and $M_{\text{OSF}}=1$ (or $M_{\text{OSF}}=2$), UAD-DC detects 38.90\% (or 24.51\%) more active users than conventional RA from Fig. \ref{fig_pmd}. Moreover, when SNR = 10 dB and $M_{\text{OSF}}=1$ (or $M_{\text{OSF}}=2$), UAD-DC detects 22.90\% (or 18.80\%) more active users than conventional RA. Fig. \ref{fig_pfa} shows the probability of misdetection against different probabilities of false alarms, where UAD-DC outperforms conventional RA at low probability of false alarm.
\begin{table}[!htbp]
	\centering
	\caption{Complexity comparisons}
	\resizebox{\columnwidth}{!}{\begin{tabular}{ccc}
		\hline
		Oversampling Factor          &
		UAD-DC &
		Conventional RA   \\\hline
		 $M_{\text{OSF}}=1$   &
		 $1.94\times10^8$ flops   &
		 $5.32\times10^7$ flops   \\
		  $M_{\text{OSF}}=2$   &
		  $7.85\times10^8$ flops   &
		  $2.13\times10^8$ flops    \\
		   $M_{\text{OSF}}=3$   &
		   $1.95\times10^9$ flops &
		   $4.79\times10^8$ flops 
		                              \\ \hline
	\end{tabular}}\label{tab_complexity}
\end{table}
\begin{figure}[!htbp]
	\centering
	\begin{minipage}{\columnwidth}
	\resizebox{\columnwidth}{!}{\begin{tikzpicture} 
	\begin{axis}[%
		hide axis,
		xmin=10,
		xmax=40,
		ymin=0,
		ymax=0.4,
		legend columns=3,
		legend style={draw=white!15!black,legend cell align=left,font=\footnotesize}
		]
		
		\addlegendimage{color=black,dashed,line width=0.8pt,mark=star,mark options={solid},mark size=3pt}
		\addlegendentry{Conv. RA w/ $M_{\text{OSF}}=1$};
		
		\addlegendimage{color=red,dashed,line width=0.8pt,mark=triangle,mark options={solid},mark size=3pt}
		\addlegendentry{Conv. RA w/ $M_{\text{OSF}}=2$};
		
		\addlegendimage{color=blue,dashed,line width=0.8pt,mark=o,mark options={solid},mark size=3pt}
		\addlegendentry{Conv. RA w/ $M_{\text{OSF}}=3$};
		
		\addlegendimage{color=black,solid,line width=0.8pt,mark=star,mark options={solid},mark size=3pt}
		\addlegendentry{UAD-DC w/ $M_{\text{OSF}}=1$};
		
		\addlegendimage{color=red,solid,line width=0.8pt,mark=triangle,mark options={solid},mark size=3pt}
		\addlegendentry{UAD-DC w/ $M_{\text{OSF}}=2$};
		
		\addlegendimage{color=blue,solid,line width=0.8pt,mark=o,mark options={solid},mark size=3pt}
		\addlegendentry{UAD-DC w/ $M_{\text{OSF}}=3$};
		
	\end{axis}
\end{tikzpicture}}
	\end{minipage}
	\begin{minipage}{0.82\columnwidth}
	\subcaptionbox{The probability of misdetection against different SNRs, where the probability of false alarm is lower than $10^{-3}$.\label{fig_pmd}}{\hspace{-0.6cm}
\begin{tikzpicture}
	
	\begin{axis}[%
		width=\columnwidth,
		height=.6\columnwidth,
		at={(0in,0in)},
		scale only axis,
		xmin=-10,
		xmax=10,
		ylabel style={font=\footnotesize},
		xlabel style={font=\footnotesize},
		xlabel={SNR (dB)},
		xtick=data,
		xmajorgrids,
		ymin=0.2,
		ymax=1,
		yminorticks=true,
		ylabel={Probability of misdetection},
		ymajorgrids,
		yminorgrids,
		axis background/.style={fill=white},
		]
		
		\addplot [color=black,dashed,line width=0.8pt,mark=star,mark options={solid},mark size=3pt]
		table[row sep=crcr]{%
			-10   0.9951\\
			-7.5  0.9738\\
			-5  0.9049\\
			-2.5  0.8022\\
			 0    0.7036\\
			 2.5  0.6272\\
			 5  0.5762\\
			 7.5  0.5447\\
			 10  0.5270\\	
		};	
		
		\addplot [color=red,dashed,line width=0.8pt,mark=triangle,mark options={solid},mark size=3pt]
		table[row sep=crcr]{%
			-10   0.9929\\
			-7.5  0.9596\\
			-5    0.8617\\
			-2.5  0.7294\\
			0     0.6137\\
			2.5   0.5319\\
			5     0.4916\\
			7.5   0.4614\\
			10 	  0.4437\\			
		};	
		
		\addplot [color=blue,dashed,line width=0.8pt,mark=o,mark options={solid},mark size=3pt]
		table[row sep=crcr]{%
			-10   0.9932\\
			-7.5  0.9559\\
			-5    0.8469\\
			-2.5  0.7140\\
			0     0.5979\\
			2.5   0.5136\\
			5     0.4717\\
			7.5   0.4418\\
			10 	  0.4212\\
		};	
		
		\addplot [color=black,solid,line width=0.8pt,mark=star,mark options={solid},mark size=3pt]
		table[row sep=crcr]{%
			-10   0.9778\\
			-7.5  0.9304\\
			-5    0.8273\\
			-2.5  0.6972\\
		     0    0.5883\\
			 2.5  0.5102\\
			 5    0.4676\\
			 7.5  0.4357\\
			 10   0.4187\\
		};	
		%
		\addplot [color=red,solid,line width=0.8pt,mark=triangle,mark options={solid},mark size=3pt]
		table[row sep=crcr]{%
			-10   0.9728\\
			-7.5  0.9113\\
			-5    0.7912\\
			-2.5  0.6418\\
			0     0.5190\\
			2.5   0.4406\\
			5     0.393\\
			7.5   0.362\\
			10 	  0.3391\\		
		};	
		
		\addplot [color=blue,solid,line width=0.8pt,mark=o,mark options={solid},mark size=3pt]
		table[row sep=crcr]{%
			-10   0.9731\\
			-7.5  0.9083\\
			-5  0.7740\\
			-2.5  0.6156\\
			0  0.5001\\
			2.5  0.4196\\
			5  0.3711\\
			7.5  0.3470\\
			10  0.3269\\
		};	
	\end{axis}
	
\end{tikzpicture}%
}
	\subcaptionbox{The probability of misdetection against different probabilities of false alarms, where SNR = 10 dB.\label{fig_pfa}}{\hspace{-0.6cm}
\begin{tikzpicture}
	
	\begin{axis}[%
		width=\columnwidth,
		height=.6\columnwidth,
		at={(0in,0in)},
		scale only axis,
		xmin=0,
		xmax=0.7,
		ylabel style={font=\footnotesize},
		xlabel style={font=\footnotesize},
		xlabel={Probability of false alarm},
		xtick={0,0.1,0.2,0.3,0.4,0.5,0.6,0.7,0.8,0.9,1},
		xticklabels={0,0.1,0.2,0.3,0.4,0.5,0.6,0.7,0.8,0.9,1},
		xmajorgrids,
		ymin=0,
		ymax=0.6,
		yminorticks=true,
		ylabel={Probability of misdetection},
		ytick={0,0.1,0.2,0.3,0.4,0.5,0.6},
		yticklabels={0,0.1,0.2,0.3,0.4,0.5,0.6},
		ymajorgrids,
		yminorgrids,
		axis background/.style={fill=white},
		]
		
		\addplot [color=black,dashed,line width=0.8pt,mark=star,mark options={solid},mark size=3pt]
		table[row sep=crcr]{%
			0         0.5270\\
			0.0288    0.4548\\
			0.0566    0.4250\\
			0.0825    0.403\\
			0.1457    0.358\\
			0.1862    0.3344\\
			0.2279    0.3176\\
			0.2722    0.3043\\
			0.3250    0.285\\
			0.3711    0.27\\
			0.4283    0.2517\\
			0.4777    0.2333\\
			0.5274    0.2167\\
			0.5796    0.1999\\ 
			0.6403    0.1844\\ 
			0.6973    0.1680\\
			0.7599    0.1590\\
		};	
		
		\addplot [color=red,dashed,line width=0.8pt,mark=triangle,mark options={solid},mark size=3pt]
		table[row sep=crcr]{%
			0         0.4437\\
			0.0116    0.4049\\
			0.0497    0.3390\\
			0.0948    0.3116\\
			0.140     0.2719\\
			0.1858    0.2538\\
			0.2261    0.2387\\
			0.2646    0.2279\\
			0.3169    0.2130\\
			0.3617    0.2030\\
			0.4249    0.1868\\
			0.4704    0.1750\\
			0.5274    0.1588\\
			0.5785    0.1439\\
			0.6344    0.1277\\
			0.6921    0.1106\\
			0.7576    0.1104\\		
		};	
		
		\addplot [color=blue,dashed,line width=0.8pt,mark=o,mark options={solid},mark size=3pt]
		table[row sep=crcr]{%
			0         0.4212\\
			0.0134    0.3821\\  
			0.0506    0.3241\\  
			0.0847    0.3050\\
			0.1407    0.2634\\
			0.1827    0.2462\\
			0.2309    0.2321\\
			0.2673    0.2224\\
			0.3177    0.2078\\
			0.3707    0.1950\\
			0.4255    0.1818\\
			0.4714    0.1679\\
			0.5242    0.1541\\
			0.5771    0.1393\\
			0.6375    0.1207\\
			0.6942    0.1077\\
			0.7571    0.1044\\
		};	
		
		\addplot [color=black,solid,line width=0.8pt,mark=star,mark options={solid},mark size=3pt]
		table[row sep=crcr]{%
			0         0.4187\\
			0.0306    0.3849\\
			0.0824    0.332\\ 
			0.1345    0.3052\\ 
			0.1896    0.2900\\
			0.2408    0.2753\\
			0.2893    0.265\\
			0.3412    0.2549\\
			0.3941    0.241\\
			0.4437    0.232\\
			0.4778    0.223\\
			0.5449    0.205\\
			0.5929    0.19\\
			0.6291    0.1810\\
			0.6956    0.16\\  
			0.7178    0.1515\\  
			0.7789    0.12\\ 
			0.8261    0.0947\\
		};	
		
		\addplot [color=red,solid,line width=0.8pt,mark=triangle,mark options={solid},mark size=3pt]
		table[row sep=crcr]{%
			0         0.3391\\
			0.0342    0.2960\\
			0.0838    0.2476\\
			0.1352    0.2272\\
			0.1894    0.2148\\
			0.2408    0.2068\\
			0.2860    0.1977\\
			0.3419    0.1858\\
			0.3937    0.1762\\
			0.4441    0.1673\\
			0.4790    0.160\\
			0.5450    0.1436\\
			0.5906    0.1352\\
			0.6239    0.1307\\
			0.6950    0.1121\\
			0.7122    0.1054\\
			0.7754    0.0848\\	
		};	
		
		\addplot [color=blue,solid,line width=0.8pt,mark=o,mark options={solid},mark size=3pt]
		table[row sep=crcr]{%
			0         0.3269\\
			0.0356    0.276\\
			0.0842    0.2327\\
			0.1351    0.2169\\
			0.1897    0.2064\\
			0.2420    0.1982\\
			0.2860    0.1880\\
			0.3416    0.1758\\
			0.3930    0.167\\
			0.4440    0.1565\\
			0.4756    0.15\\
			0.5434    0.1368\\
			0.5886    0.1301\\
			0.6209    0.1226\\
			0.6950    0.1051\\ 
			0.7743    0.0771\\
		};	
	\end{axis}
\end{tikzpicture}%
}
	\end{minipage}
	\caption{Performance comparisons between conventional RA and UAD-DC.}\label{fig_peralg}
\end{figure}

Fig. \ref{fig_cvxfilter} compares the impulse responses, frequency responses and BCRBs of various pulse shaping filters under the same spectral constraint, represented by the dashed lines. Under this spectral constraint, the optimized filter outperforms the RRC filter with $\beta=0.4$ and the Gaussian filter with $\sigma_\text{G}^2=0.49$ in terms of the BCRBs shown in Fig. \ref{fig_cvxcrb}. 
\begin{figure}[!htbp]
	\centering
	\begin{minipage}{\columnwidth}
	\resizebox{\columnwidth}{!}{\begin{tikzpicture} 
	\begin{axis}[%
		hide axis,
		xmin=10,
		xmax=40,
		ymin=0,
		ymax=0.4,
		legend columns=4,
		legend style={draw=white!15!black,legend cell align=left}
		]
		
		\addlegendimage{blue,line width=1.0pt}
		\addlegendentry{Gaussian w/ $\sigma_\text{G}^2=0.49$};
		
		\addlegendimage{black,line width=1.0pt}
		\addlegendentry{RRC w/ $\beta=0.4$};
		
		\addlegendimage{red,line width=1.0pt}
		\addlegendentry{Optimized filter};
		
		
	\end{axis}
\end{tikzpicture}}
	\end{minipage}\vfill
	\begin{minipage}{0.82\columnwidth}
		\centering
		\subcaptionbox{Impulse responses of $z(t)$.\label{fig_cvxtime}}{\hspace{-0.7cm}\begin{tikzpicture}
	
	\begin{axis}[%
		width=\columnwidth,
		height=.5\columnwidth,
		at={(0in,0in)},
		scale only axis,
		xmin=1,
		xmax=73,
		xtick={1,13,25,37,49,61,73},
		xticklabels={$-3T_s$,$-2T_s$,$-T_s$,$0$,$T_s$,$2T_s$,$3T_s$},
		ymin=-0.4,
		ymax=1,
		ylabel style={font=\footnotesize},
		xlabel style={font=\footnotesize},
		ylabel={$z(t)$},
		xlabel={$t$},
		axis background/.style={fill=white},
		xmajorgrids,
		ymajorgrids,
		]
		
		\addplot [color=black,solid,line width=1pt]
		table[row sep=crcr]{%
			1	6.62539340231493e-18\\
			2	0.00550392756109055\\
			3	0.0124018210454291\\
			4	0.0202711103032212\\
			5	0.028474241198657\\
			6	0.0361853277298916\\
			7	0.0424413181578388\\
			8	0.0462161015161893\\
			9	0.0465135061951994\\
			10	0.0424728025400825\\
			11	0.033478359489341\\
			12	0.0192637464638169\\
			13	-2.02159439711661e-17\\
			14	-0.0236418882698099\\
			15	-0.0504629164014464\\
			16	-0.0787489067017016\\
			17	-0.106327715545631\\
			18	-0.130670772985305\\
			19	-0.149035097436889\\
			20	-0.158638167495007\\
			21	-0.15685465516374\\
			22	-0.14142135623731\\
			23	-0.110634971383463\\
			24	-0.0635268985295551\\
			25	3.3461148423395e-17\\
			26	0.0790855970274923\\
			27	0.171887338539247\\
			28	0.275621173455956\\
			29	0.38666917812935\\
			30	0.500749379580277\\
			31	0.613138350952957\\
			32	0.7189328054709\\
			33	0.813333027401414\\
			34	0.891928853341733\\
			35	0.950968256723325\\
			36	0.987589474516573\\
			37	1\\
			38	0.987589474516573\\
			39	0.950968256723325\\
			40	0.891928853341733\\
			41	0.813333027401414\\
			42	0.7189328054709\\
			43	0.613138350952957\\
			44	0.500749379580277\\
			45	0.38666917812935\\
			46	0.275621173455956\\
			47	0.171887338539247\\
			48	0.0790855970274923\\
			49	3.3461148423395e-17\\
			50	-0.0635268985295551\\
			51	-0.110634971383463\\
			52	-0.14142135623731\\
			53	-0.15685465516374\\
			54	-0.158638167495007\\
			55	-0.149035097436889\\
			56	-0.130670772985305\\
			57	-0.106327715545631\\
			58	-0.0787489067017016\\
			59	-0.0504629164014464\\
			60	-0.0236418882698099\\
			61	-2.02159439711661e-17\\
			62	0.0192637464638169\\
			63	0.033478359489341\\
			64	0.0424728025400825\\
			65	0.0465135061951994\\
			66	0.0462161015161893\\
			67	0.0424413181578388\\
			68	0.0361853277298916\\
			69	0.028474241198657\\
			70	0.0202711103032212\\
			71	0.0124018210454291\\
			72	0.00550392756109055\\
			73	6.62539340231493e-18\\
		};
		
		\addplot [color=blue,solid,line width=1pt]
		table[row sep=crcr]{%
			1	8.51493685255308e-05\\
			2	0.000142279889876355\\
			3	0.000234328102543632\\
			4	0.00038038577549304\\
			5	0.000608615991391358\\
			6	0.000959802208478728\\
			7	0.0014918991985748\\
			8	0.00228568615860116\\
			9	0.00345154162992133\\
			10	0.00513723039231824\\
			11	0.00753640771941914\\
			12	0.010897306004318\\
			13	0.0155307821057586\\
			14	0.0218165928844485\\
			15	0.0302064681556683\\
			16	0.0412223163384991\\
			17	0.0554477918629579\\
			18	0.0735115433516366\\
			19	0.0960608106038349\\
			20	0.123724690111655\\
			21	0.157067352582185\\
			22	0.196532733132461\\
			23	0.242383629994882\\
			24	0.294639587633964\\
			25	0.353019204566113\\
			26	0.416893368942807\\
			27	0.485256167541429\\
			28	0.556719664513747\\
			29	0.629537320023702\\
			30	0.701658549490564\\
			31	0.770813980532833\\
			32	0.834627645813719\\
			33	0.890749054895674\\
			34	0.93699526395958\\
			35	0.971491136351545\\
			36	0.992795295749356\\
			37	1\\
			38	0.992795295749356\\
			39	0.971491136351545\\
			40	0.93699526395958\\
			41	0.890749054895674\\
			42	0.834627645813719\\
			43	0.770813980532833\\
			44	0.701658549490564\\
			45	0.629537320023702\\
			46	0.556719664513747\\
			47	0.485256167541429\\
			48	0.416893368942807\\
			49	0.353019204566113\\
			50	0.294639587633964\\
			51	0.242383629994882\\
			52	0.196532733132461\\
			53	0.157067352582185\\
			54	0.123724690111655\\
			55	0.0960608106038349\\
			56	0.0735115433516366\\
			57	0.0554477918629579\\
			58	0.0412223163384991\\
			59	0.0302064681556683\\
			60	0.0218165928844485\\
			61	0.0155307821057586\\
			62	0.010897306004318\\
			63	0.00753640771941914\\
			64	0.00513723039231824\\
			65	0.00345154162992133\\
			66	0.00228568615860116\\
			67	0.0014918991985748\\
			68	0.000959802208478728\\
			69	0.000608615991391358\\
			70	0.00038038577549304\\
			71	0.000234328102543632\\
			72	0.000142279889876355\\
			73	8.51493685255308e-05\\
		};

		\addplot [color=red,solid,line width=1pt]
		table[row sep=crcr]{%
			1	0.00448437560819185\\
			2	0.0108022236319071\\
			3	0.0163721704469713\\
			4	0.0313253209240146\\
			5	0.0486321798510757\\
			6	0.0610011128410071\\
			7	0.0846267894983871\\
			8	0.108569921764068\\
			9	0.120537737998844\\
			10	0.139649423751704\\
			11	0.153176046615915\\
			12	0.147712639108088\\
			13	0.142134408230272\\
			14	0.125320072519813\\
			15	0.0897359417475521\\
			16	0.0505732669444751\\
			17	8.70963453744987e-05\\
			18	-0.0577333415717145\\
			19	-0.114157675128573\\
			20	-0.173126225193942\\
			21	-0.219767532595724\\
			22	-0.253266459831403\\
			23	-0.274977370119636\\
			24	-0.267624095416758\\
			25	-0.238114565724953\\
			26	-0.186050242843476\\
			27	-0.104020603361651\\
			28	-0.00448528502407951\\
			29	0.117156013605782\\
			30	0.249949565624331\\
			31	0.379973992403224\\
			32	0.521661119636239\\
			33	0.644377645171849\\
			34	0.737164063260546\\
			35	0.828853184565468\\
			36	0.877726115702138\\
			37	1.00000000083948\\
			38	0.877726115702138\\
			39	0.828853184565468\\
			40	0.737164063260546\\
			41	0.644377645171849\\
			42	0.521661119636239\\
			43	0.379973992403224\\
			44	0.249949565624331\\
			45	0.117156013605782\\
			46	-0.00448528502407951\\
			47	-0.104020603361651\\
			48	-0.186050242843476\\
			49	-0.238114565724953\\
			50	-0.267624095416758\\
			51	-0.274977370119636\\
			52	-0.253266459831403\\
			53	-0.219767532595724\\
			54	-0.173126225193942\\
			55	-0.114157675128573\\
			56	-0.0577333415717145\\
			57	8.70963453744987e-05\\
			58	0.0505732669444751\\
			59	0.0897359417475521\\
			60	0.125320072519813\\
			61	0.142134408230272\\
			62	0.147712639108088\\
			63	0.153176046615915\\
			64	0.139649423751704\\
			65	0.120537737998844\\
			66	0.108569921764068\\
			67	0.0846267894983871\\
			68	0.0610011128410071\\
			69	0.0486321798510757\\
			70	0.0313253209240146\\
			71	0.0163721704469713\\
			72	0.0108022236319071\\
			73	0.00448437560819185\\
		};
	\end{axis}
\end{tikzpicture}%
}
		\subcaptionbox{Frequency responses of $z(t)$.\label{fig_cvxfre}}{\hspace{-0.7cm}\input{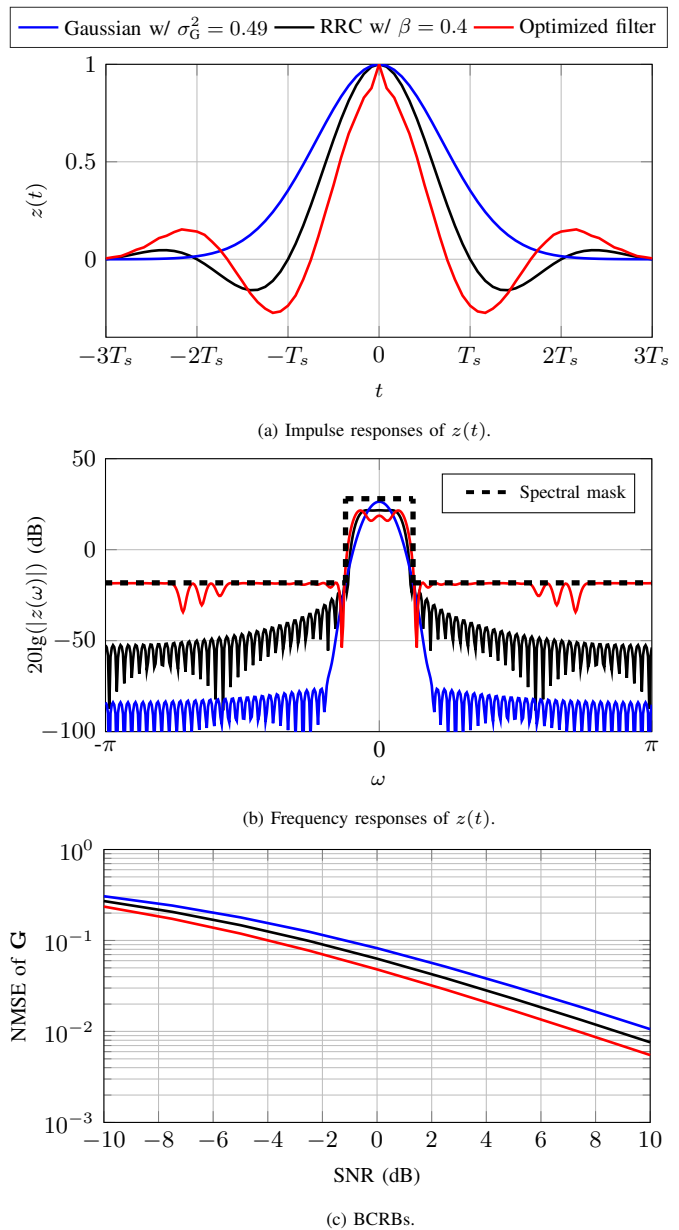}}
		\subcaptionbox{BCRBs.\label{fig_cvxcrb}}{\hspace{-0.9cm}\begin{tikzpicture}
	
	\begin{axis}[%
		width=\columnwidth,
		height=.5\columnwidth,
		at={(0in,0in)},
		scale only axis,
		xmin=-10,
		xmax=10,
		xlabel={SNR (dB)},
		xtick={-10,-8,-6,-4,-2,0,2,4,6,8,10},
		xmajorgrids,
		ymode=log,
		ymin=0.001,
		ymax=1,
		yminorticks=true,
		ylabel style={font=\footnotesize},
		xlabel style={font=\footnotesize},
		ylabel={NMSE of $\mathbf{G}$},
		axis background/.style={fill=white},
		ymajorgrids,
		yminorgrids,
		]
		
		\addplot [color=red, solid,line width=1pt]
		table[row sep=crcr]{%
			-10	  0.2358\\
			-7.5  0.1730\\
			-5    0.1191\\
			-2.5  0.0775\\
			0     0.0481\\
			2.5   0.0289\\
			5     0.0169\\
			7.5   0.0097\\
			10    0.0055\\
		};
		
		\addplot [color=black, solid,line width=1pt]
		table[row sep=crcr]{%
			-10	  0.2712\\
			-7.5  0.2064\\
			-5    0.1473\\
			-2.5  0.0990\\
			0     0.0631\\
			2.5   0.0386\\
			5     0.0229\\
			7.5   0.0133\\
			10    0.0076\\
		};
		
		\addplot [color=blue, solid,line width=1pt]
		table[row sep=crcr]{%
			-10	  0.3065\\
			-7.5  0.2421\\
			-5    0.1794\\
			-2.5  0.1250\\
			0     0.0823\\
			2.5   0.0516\\
			5     0.0312\\
			7.5   0.0184\\
			10    0.0106\\
		};
	
	\end{axis}
\end{tikzpicture}%
}
	\end{minipage}
	\caption{Performance comparisons by using different pulse shaping filters, where $z(t)$ is the convolution of the pulse shaping and its corresponding matched filter.}\label{fig_cvxfilter}
\end{figure}
Fig. \ref{fig_cvxperformance} compares the detection performance of the above filters, where the optimized filter outperforms the RRC filter with $\beta=0.4$ and the Gaussian filter with $\sigma_\text{G}^2=0.49$. For example, when SNR = 0 dB and $M_{\text{OSF}}=2$ (or $M_{\text{OSF}}=3$), the system with optimized filter can detect 6.88\% (or 9.04\%) more active users than the system with RRC filter, and 33.26\% (or 39.81\%) more active users than the system with Gaussian filter. Moreover, when SNR = 10 dB and $M_{\text{OSF}}=2$ (or $M_{\text{OSF}}=3$), the system with optimized filter can detect 3.46\% (or 7.61\%) more active users than the system with RRC filter, and 20.60\% (or 24.86\%) more active users than the system with Gaussian filter. Note that when the system is sampled at the Nyquist rate, namely $M_{\text{OSF}}=1$, the RRC filter with $\beta=0.4$ achieves the best estimation performance, as no ISI occurs. 
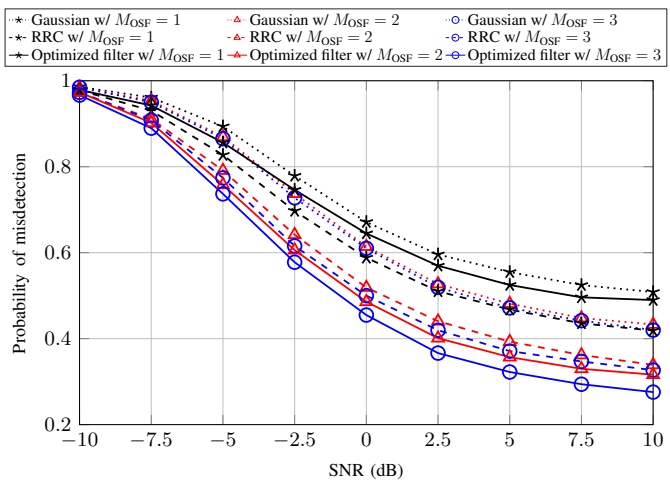
\begin{figure}[!htbp]
	\centering
	\begin{minipage}{\columnwidth}
			\resizebox{\columnwidth}{!}{\begin{tikzpicture} 
	\begin{axis}[%
		hide axis,
		xmin=10,
		xmax=40,
		ymin=0,
		ymax=0.4,
		legend columns=3,
		legend style={draw=white!15!black,legend cell align=left,font=\Large}
		]
		
		\addlegendimage{color=black,dotted,line width=0.8pt,mark=star,mark options={solid},mark size=3pt}
		\addlegendentry{Gaussian w/ $M_{\text{OSF}}=1$};
		
		\addlegendimage{color=red,dotted,line width=0.8pt,mark=triangle,mark options={solid},mark size=3pt}
		\addlegendentry{Gaussian w/ $M_{\text{OSF}}=2$};
		
		\addlegendimage{color=blue,dotted,line width=0.8pt,mark=o,mark options={solid},mark size=3pt}
		\addlegendentry{Gaussian w/ $M_{\text{OSF}}=3$};
		
		\addlegendimage{color=black,dashed,line width=0.8pt,mark=star,mark options={solid},mark size=3pt}
		\addlegendentry{RRC w/ $M_{\text{OSF}}=1$};
		
		\addlegendimage{color=red,dashed,line width=0.8pt,mark=triangle,mark options={solid},mark size=3pt}
		\addlegendentry{RRC w/ $M_{\text{OSF}}=2$};
		
		\addlegendimage{color=blue,dashed,line width=0.8pt,mark=o,mark options={solid},mark size=3pt}
		\addlegendentry{RRC w/ $M_{\text{OSF}}=3$};
		
		\addlegendimage{color=black,solid,line width=0.8pt,mark=star,mark options={solid},mark size=3pt}
		\addlegendentry{Optimized filter w/ $M_{\text{OSF}}=1$};
		
		\addlegendimage{color=red,solid,line width=0.8pt,mark=triangle,mark options={solid},mark size=3pt}
		\addlegendentry{Optimized filter w/ $M_{\text{OSF}}=2$};
		
		\addlegendimage{color=blue,solid,line width=0.8pt,mark=o,mark options={solid},mark size=3pt}
		\addlegendentry{Optimized filter w/ $M_{\text{OSF}}=3$};
		
%
%
		
	\end{axis}
\end{tikzpicture}}
	\end{minipage}\vfil
	\begin{minipage}{\columnwidth}
		\centering
			\resizebox{\columnwidth}{!}{\begin{tikzpicture}
	
	\begin{axis}[%
		width=\columnwidth,
		height=.6\columnwidth,
		at={(0in,0in)},
		scale only axis,
		xmin=-10,
		xmax=10,
		xlabel={SNR (dB)},
		xtick=data,
		xmajorgrids,
		ymin=0.2,
		ymax=1,
		yminorticks=true,
		ylabel={Probability of misdetection},
		ylabel style={font=\footnotesize},
		xlabel style={font=\footnotesize},
		ymajorgrids,
		yminorgrids,
		axis background/.style={fill=white},
		]
		
		\addplot [color=black,dotted,line width=0.8pt,mark=star,mark options={solid},mark size=3pt]
		table[row sep=crcr]{%
			-10   0.9858\\
			-7.5  0.9609\\
			-5    0.8937\\
			-2.5  0.7788\\
			0     0.6713\\
			2.5   0.5958\\
			5     0.5548\\
			7.5   0.5247\\
			10 	  0.5084\\				
		};	
		
		\addplot [color=red,dotted,line width=0.8pt,mark=triangle,mark options={solid},mark size=3pt]
		table[row sep=crcr]{%
			-10   0.9837\\
			-7.5  0.9534\\
			-5    0.8696\\
			-2.5  0.7373\\
			0     0.6142\\
			2.5   0.5272\\
			5     0.4814\\
			7.5   0.4480\\
			10 	  0.4330\\		
		};	
		
		\addplot [color=blue,dotted,line width=0.8pt,mark=o,mark options={solid},mark size=3pt]
		table[row sep=crcr]{%
			-10   0.9848\\
			-7.5  0.9500\\
			-5    0.8653\\
			-2.5  0.7277\\
			0     0.6101\\
			2.5   0.5193\\
			5     0.4716\\
			7.5   0.4419\\
			10 	  0.4199\\
		};	
		
		\addplot [color=black,dashed,line width=0.8pt,mark=star,mark options={solid},mark size=3pt]
		table[row sep=crcr]{%
			-10   0.9778\\
			-7.5  0.9304\\
			-5    0.8273\\
			-2.5  0.6972\\
			0    0.5883\\
			2.5  0.5102\\
			5    0.4676\\
			7.5  0.4357\\
			10   0.4187\\		
		};	
				
		\addplot [color=red,dashed,line width=0.8pt,mark=triangle,mark options={solid},mark size=3pt]
		table[row sep=crcr]{%
			-10   0.9728\\
			-7.5  0.9113\\
			-5    0.7912\\
			-2.5  0.6418\\
			0     0.5190\\
			2.5   0.4406\\
			5     0.393\\
			7.5   0.362\\
			10 	  0.3391\\			
		};	
		
		\addplot [color=blue,dashed,line width=0.8pt,mark=o,mark options={solid},mark size=3pt]
		table[row sep=crcr]{%
			-10   0.9731\\
			-7.5  0.9083\\
			-5  0.7740\\
			-2.5  0.6156\\
			0  0.5001\\
			2.5  0.4196\\
			5  0.3711\\
			7.5  0.3470\\
			10  0.3269\\
		};	
		
		\addplot [color=black,solid,line width=0.8pt,mark=star,mark options={solid},mark size=3pt]
		table[row sep=crcr]{%
			-10   0.9787\\
			-7.5  0.9424\\
			-5    0.8562\\
			-2.5  0.7461\\
			0     0.6448\\
			2.5   0.5697\\
			5     0.5248\\
			7.5   0.4961\\
			10    0.4898\\
		};	
		
		\addplot [color=red,solid,line width=0.8pt,mark=triangle,mark options={solid},mark size=3pt]
		table[row sep=crcr]{%
			-10   0.9733\\
			-7.5  0.9021\\
			-5    0.7579\\
			-2.5  0.6060\\
			0     0.4859\\
			2.5   0.4011\\
			5     0.3570\\
			7.5   0.3299\\
			10    0.3162\\
		};	
		
		\addplot [color=blue,solid,line width=0.8pt,mark=o,mark options={solid},mark size=3pt]
		table[row sep=crcr]{%
			-10   0.9666\\
			-7.5  0.8897\\
			-5    0.7368\\
			-2.5  0.5778\\
			0     0.4549\\
			2.5   0.3664\\
			5     0.3224\\
			7.5   0.2940\\
			10    0.2757\\			
		};	
	\end{axis}
	
\end{tikzpicture}%
}
	\end{minipage}
	\caption{Performance comparisons of different pulse shaping filters by using UAD-DC, where the probability of false alarm is lower than $10^{-3}$, the RRC filter is with $\beta=0.4$, and the Gaussian filter is with $\sigma_\text{G}^2=0.49$.}\label{fig_cvxperformance}
\end{figure}

Finally, we present the user activity detection performance of different pulse shaping filters against different numbers of active users, ranging from 100 to 500 in Fig. \ref{fig_cvxperformanceact}, where the successful detection ratio is defined as the number of detected active users over the number of active users. It can be seen that UAD-DC with optimized filter achieves the best performance. For example, when there are 500 active users and $M_{\text{OSF}}=2$ (or $M_{\text{OSF}}=3$), UAD-DC with optimized filter can detect 6.63\% (or 15.83\%) more active users than UAD-DC with RRC filter. This verifies the effectiveness of the proposed filter optimization algorithm. Furthermore, we also present the performance of the conventional RA with RRC filter as our benchmark. When there are 500 active users and $M_{\text{OSF}}=2$ (or $M_{\text{OSF}}=3$), the proposed UAD-DC with optimized filter can detect 16.49\% (or 26.48\%) more active users than the benchmark. This verifies the effectiveness of the proposed UAD-DC with optimized filter in comparison to the conventional RA with RRC filter.
\begin{figure}[!htbp]
	\centering
	\begin{minipage}{\columnwidth}
			\resizebox{\columnwidth}{!}{\begin{tikzpicture} 
	\begin{axis}[%
		hide axis,
		xmin=10,
		xmax=40,
		ymin=0,
		ymax=0.4,
		legend columns=2,
		legend style={draw=white!15!black,legend cell align=left}
		]
		
		\addlegendimage{color=red,dash dot,line width=0.8pt,mark=triangle,mark options={solid},mark size=3pt}
		\addlegendentry{Conv. RA w/ RRC \& $M_{\text{OSF}}=2$};
		
		\addlegendimage{color=blue,dash dot,line width=0.8pt,mark=o,mark options={solid},mark size=3pt}
		\addlegendentry{Conv. RA w/ RRC \& $M_{\text{OSF}}=3$};
		
		\addlegendimage{color=red,dotted,line width=0.8pt,mark=triangle,mark options={solid},mark size=3pt}
		\addlegendentry{UAD-DC w/ Gaussian \& $M_{\text{OSF}}=2$};
		
		\addlegendimage{color=blue,dotted,line width=0.8pt,mark=o,mark options={solid},mark size=3pt}
		\addlegendentry{UAD-DC w/ Gaussian \& $M_{\text{OSF}}=3$};
		
		\addlegendimage{color=red,dashed,line width=1pt,mark=triangle,mark options={solid},mark size=3pt}
		\addlegendentry{UAD-DC w/ RRC \& $M_{\text{OSF}}=2$};
		
		\addlegendimage{color=blue,dashed,line width=1pt,mark=o,mark options={solid},mark size=3pt}
		\addlegendentry{UAD-DC w/ RRC \& $M_{\text{OSF}}=3$};
		
		\addlegendimage{color=red,solid,line width=1pt,mark=triangle,mark options={solid},mark size=3pt}
		\addlegendentry{UAD-DC w/ Optimized filter \& $M_{\text{OSF}}=2$};
		
		
		\addlegendimage{color=blue,solid,line width=1pt,mark=o,mark options={solid},mark size=3pt}
		\addlegendentry{UAD-DC w/ Optimized filter \& $M_{\text{OSF}}=3$};
		
		
	\end{axis}
\end{tikzpicture}}
	\end{minipage}\vfil
	\begin{minipage}{\columnwidth}
		\centering
			\resizebox{\columnwidth}{!}{\begin{tikzpicture}
	
	\begin{axis}[%
		width=\columnwidth,
		height=.6\columnwidth,
		at={(0in,0in)},
		scale only axis,
		xmin=100,
		xmax=500,
		xlabel style={font=\footnotesize},
		ylabel style={font=\footnotesize},
		xlabel={Number of active users},
		xtick=data,
		xmajorgrids,
		ymin=0.2,
		ymax=1,
		yminorticks=true,
		ylabel style={at={(axis description cs:-0.06,0.5)}},
		ylabel={Successful detection ratio},
		ymajorgrids,
		yminorgrids,
		axis background/.style={fill=white},
		]
		
		\addplot [color=red,dashed,line width=0.8pt,mark=triangle,mark options={solid},mark size=3pt]
		table[row sep=crcr]{%
			100   0.9067\\
			150   0.8260\\
			200   0.7365\\
			250   0.6975\\
			300   0.6609\\
			350   0.5861\\
			400   0.515\\ 
			450   0.4382\\
			500   0.3829\\		
		};	
		
		\addplot [color=blue,dashed,line width=0.8pt,mark=o,mark options={solid},mark size=3pt]
		table[row sep=crcr]{%
			100   0.9110\\
			150   0.8327\\
			200   0.7570\\
			250   0.7112\\
			300   0.6731\\
			350   0.6\\
			400   0.535\\ 
			450   0.465\\
			500   0.405\\ 
		};	
		
		\addplot [color=red,solid,line width=0.8pt,mark=triangle,mark options={solid},mark size=3pt]
		table[row sep=crcr]{%
			100   0.9277\\
			150   0.8502\\
			200   0.7890\\
			250   0.7339\\
			300   0.6838\\
			350   0.610\\
			400   0.54\\ 
			450   0.4686\\
			500   0.4083\\
		};	
		
		\addplot [color=blue,solid,line width=0.8pt,mark=o,mark options={solid},mark size=3pt]
		table[row sep=crcr]{%
			100   0.9370\\
			150   0.8698\\
			200   0.8078\\
			250   0.7580\\
			300   0.7243\\
			350   0.655\\ 
			400   0.59\\ 
			450   0.5345\\
			500   0.4691\\		
		};	
		
		\addplot [color=red,dash dot,line width=0.8pt,mark=triangle,mark options={solid},mark size=3pt]
		table[row sep=crcr]{%
			100   0.8243\\
			150   0.7327\\
			200   0.6525\\
			250   0.5863\\
			300   0.5563\\
			350   0.4726\\
			400   0.44\\
			450   0.3850\\
			500   0.3505\\	
		};	
		
		\addplot [color=blue,dash dot,line width=0.8pt,mark=o,mark options={solid},mark size=3pt]
		table[row sep=crcr]{%
			100   0.8540\\
			150   0.7589\\
			200   0.6702\\
			250   0.6017\\
			300   0.5788\\
			350   0.4946\\  
			400   0.46\\
			450   0.4070\\
			500   0.3709\\			
		};	
		
		\addplot [color=red,dotted,line width=0.8pt,mark=triangle,mark options={solid},mark size=3pt]
		table[row sep=crcr]{%
			100   0.8640\\
			150   0.755\\ 
			200   0.6668\\
			250   0.6127\\
			300   0.5670\\
			350   0.485\\ 
			400   0.4219\\ 
			450   0.3727\\
			500   0.3168\\	
		};	
		
		\addplot [color=blue,dotted,line width=0.8pt,mark=o,mark options={solid},mark size=3pt]
		table[row sep=crcr]{%
			100   0.8860\\
			150   0.77\\
			200   0.6878\\
			250   0.6229\\
			300   0.5801\\
			350   0.5014\\
			400   0.439\\
			450   0.3857\\
			500   0.3357\\	
		};	
	\end{axis}
	
\end{tikzpicture}%
}
	\end{minipage}
	\caption{Performance comparisons of different pulse shaping filters by using UAD-DC, where the probability of false alarm is lower than $10^{-3}$, the RRC filter is with $\beta=0.4$, and the Gaussian filter is with $\sigma_\text{G}^2=0.49$, and SNR = 10 dB.}\label{fig_cvxperformanceact}
\end{figure}
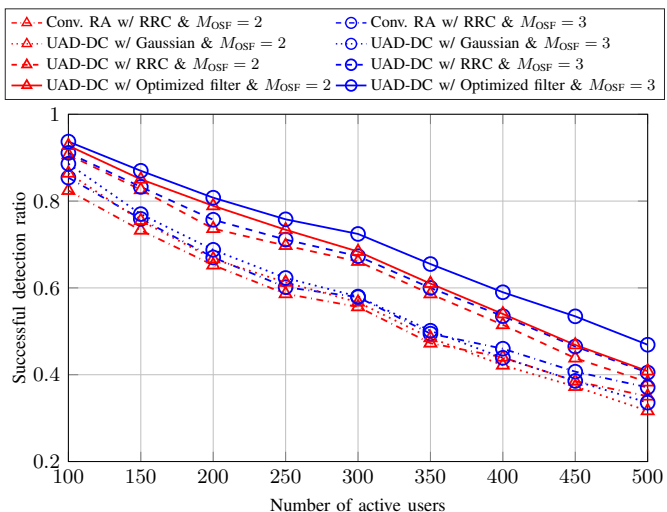

\section{Conclusions}
In this work, we proposed a sliding-window-based UAD-DC algorithm for asynchronous massive random access systems. In each window, the time delays of active users were estimated by the EM method iteratively. Based on the estimated time delays, a modified Turbo-CS-MMV algorithm was presented for detecting the active users with considering the noise correlations due to oversampling. Furthermore, a pulse shaping filter optimization algorithm was proposed to improve the performance of user activity detection with the objective of minimizing the BCRB under the constraint of limited spectral bandwidth. Numerical results demonstrated the efficacy of the proposed algorithm with the optimized filter in terms of the probability of misdetection and the NMSE of channel matrix.

The main complexity of the proposed UAD-DC algorithm lies in the matrix inversion in the modified Turbo-CS-MMV algorithm and the greedy search in the delay-calibration algorithm. Exploring advanced techniques to reduce the computational cost while maintaining the same performance will be an interesting topic to pursue in our future research.


\bibliographystyle{IEEEtran}
\bibliography{ref}

\end{document}